\documentclass[journal]{IEEEtran}
\usepackage{amsmath,amsfonts}
\usepackage{algorithm}
\usepackage{algpseudocode}
\usepackage[dvipsnames]{xcolor}
\everymath{\displaystyle}

\usepackage{array}
\usepackage[caption=false,font=normalsize,labelfont=sf,textfont=sf]{subfig}
\usepackage{textcomp}
\usepackage{stfloats}
\usepackage{url}
\usepackage{verbatim}
\usepackage{graphicx}
\usepackage{graphics}
\usepackage{stackengine}
\usepackage{setspace}
\usepackage{amsthm}
\usepackage{adjustbox}
\usepackage{subfig}
\usepackage{xcolor}

\usepackage{amsmath}
\usepackage{amssymb}
\usepackage{amsthm}

\newtheorem{theorem}{Theorem}

\newtheorem{lemma}[theorem]{Lemma}

\usepackage{cite}
\usepackage{blindtext}
\begin{document}

\title{Safe Deep Reinforcement Learning for Resource Allocation with Peak Age of Information Violation Guarantees}

\author{Berire Gunes Reyhan and Sinem Coleri,~\IEEEmembership{Fellow,~IEEE}
        \thanks{B. Gunes Reyhan and S. Coleri are with the Department
of Electrical and Electronics Engineering, Koc University, Istanbul, e-mail: \{beriregunes22, scoleri\}@ku.edu.tr. 
Sinem Coleri acknowledges the support of the Scientific and Technological
Research Council of Turkey 2247-A National Leaders Research Grant
\#121C314. }
}

\markboth{}%
{}

\maketitle

\begin{abstract}
In Wireless Networked Control Systems (WNCSs), control and communication systems must be co-designed due to their strong interdependence. This paper presents a novel optimization theory-based safe deep reinforcement learning (DRL) framework for ultra-reliable WNCSs, ensuring constraint satisfaction while optimizing performance, for the first time in the literature. The approach minimizes power consumption under key constraints, including Peak Age of Information (PAoI) violation probability, transmit power, and schedulability in the finite blocklength regime. PAoI violation probability is uniquely derived by combining stochastic maximum allowable transfer interval (MATI) and maximum allowable packet delay (MAD) constraints in a multi-sensor network. The framework consists of two stages: optimization theory and safe DRL. The first stage derives optimality conditions to establish mathematical relationships among variables, simplifying and decomposing the problem. The second stage employs a safe DRL model where a teacher-student framework guides the DRL agent (student). The control mechanism (teacher) evaluates compliance with system constraints and suggests the nearest feasible action when needed. Extensive simulations show that the proposed framework outperforms rule-based and other optimization theory based DRL benchmarks, achieving faster convergence, higher rewards, and greater stability.
\end{abstract}


\setlength{\parindent}{10pt}

\section{Introduction}
\IEEEPARstart{W}{ireless} Networked Control Systems (WNCSs) comprise spatially distributed sensors, actuators, and controllers that communicate via wireless networks \cite{wncs_survey_2018}. Utilizing wireless communication within control systems leads to cost-effective and adaptable network architectures, reducing expenses associated with system component installation, modification, and upgrades when compared to their wired counterparts. As a result, WNCSs have been successfully deployed in a wide range of applications, including automotive electrical systems \cite{ wncs_vehicular_sadi}, avionics control systems \cite{WinNT}, building automation systems \cite{wncs_building_2010}, and industrial automation systems \cite{wncs_industrial_2009}. The main challenge in designing the communication system for a WNCS lies in maintaining control system performance and stability despite the inherent unreliability of wireless transmissions, transmission delays, the shared wireless medium, and the limited battery resources of sensor nodes.

The performance and stability of WNCSs are commonly evaluated using key metrics, including stochastic maximum allowable transfer interval (MATI), maximum allowable packet delay (MAD) \cite{sadi2014minimum, sadi2017joint, hamida_amir_2024}, Age of Information (AoI) \cite{aoi_definition_2016, aoi_discrete_2021}, and Peak Age of Information (PAoI) \cite{aoi_discrete_2021, joint_paoi_fbc_2021, urllc_pai_devassy_2019, paoi_violation_2019, ra_fb_pao_2023, aoi_bound_ailing_2023}. Stochastic MATI entails maintaining the time interval between consecutive state vector reports from sensor nodes to the controller below a specified threshold with a given probability, while MAD denotes the maximum permissible delay for packets transmitted from sensor nodes to the controller \cite{sadi2014minimum, sadi2017joint, hamida_amir_2024}. AoI quantifies the freshness of information at the destination node, with PAoI capturing the maximum AoI immediately before the reception of an update \cite{aoi_definition_2016}. A crucial performance metric in this context is the probability that the PAoI exceeds a predetermined threshold, referred to as the PAoI violation probability. This metric is essential for designing systems that guarantee the freshness of the information at the destination at any given time. 

The PAoI violation probability has been formulated in the literature for both single-node and multi-node systems. \cite{aoi_discrete_2021} examines the distributions of AoI and PAoI in a source-destination communication link for a wide class of discrete time status update systems. In addition, \cite{joint_paoi_fbc_2021} computes PAoI violation probability indirectly by deriving an upper-bound using the Mellin transform for point-to-point communication in the finite blocklength (FBL) regime. Moreover, \cite{urllc_pai_devassy_2019} characterizes PAoI violation probability through the probability-generating function (PGF) of the steady-state PAoI for a single-node system using FBL theory. Extending this analysis to multi-node systems, \cite{paoi_violation_2019} and \cite{ra_fb_pao_2023} also employ PGF to obtain PAoI violation probability. Furthermore, \cite{aoi_bound_ailing_2023} derives an upper bound on PAoI violation probability in the FBL regime for a system comprising multiple unmanned aerial vehicles (UAVs) by applying the Mellin transform. However, none of these studies directly derive the PAoI violation probability for a multi-node system satisfying ultra-reliable low-latency communication (URLLC) requirements within the FBL regime.

WNCS optimization problems aim to determine the optimal values for packet error probability, sampling time, and network scheduling parameters by utilizing various objective and constraint functions. These problems are often solved through either model-based techniques \cite{sadi2014minimum, sadi2017joint,joint_paoi_fbc_2021, ra_fb_pao_2023, ra_rsma_wncs_2024, AoL_wncs, aoi_bound_ailing_2023} or machine learning-based approaches \cite{dl_wncs_2024, drl_wncs_2022, hamida_amir_2024}. Model-based approaches typically involve deriving mathematical models that capture the system behavior and then solving optimization problems using established analytical techniques. In \cite{sadi2014minimum}, the optimization of energy-efficient data transmission in a WNCS is addressed, aiming to minimize power consumption while satisfying constraints related to system stability, schedulability, and maximum transmit power. The authors use relaxation techniques to simplify and solve the complex Integer Programming problem. In \cite{sadi2017joint}, the focus shifts to jointly optimizing energy consumption and control system performance, with the goal of minimizing energy usage while ensuring high reliability and strict delay requirements. The solution employs a combination of optimality conditions, smart enumeration, and heuristic algorithms to achieve near-optimal performance. \cite{joint_paoi_fbc_2021} jointly optimizes PAoI and delay-bound violation probabilities under FBL constraints, employing stochastic network calculus and convex optimization techniques to balance trade-offs between these metrics. In \cite{ra_fb_pao_2023}, resource allocation is optimized to minimize delay and Peak AoI violation probabilities under packet error constraints in URLLC systems using a combination of analytical derivations and numerical optimization. \cite{aoi_bound_ailing_2023} optimizes PAoI violation probability under the constraints of transmit power, UAV trajectories, and decoding error probability, using an iterative approach involving convex optimization. In \cite{ra_rsma_wncs_2024}, resource allocation and user pairing are optimized in rate splitting multiple access (RSMA)-based WNCSs to maximize sum rate while ensuring control stability under imperfect channel state information (CSI), using semi-definite programming relaxation, successive convex approximation, and hypergraph game theory. \cite{AoL_wncs} analyzes Peak Age of Loop (PAoL) in WNCSs under the FBL regime, with closed-form expressions for its average, variance, and outage probability. A PAoL-oriented blocklength adaptation scheme with retransmission is proposed, optimizing uplink (UL) and downlink (DL) blocklengths using convex optimization and Newton’s method while ensuring stability with a Linear Quadratic Regulator (LQR) controller.  However, the high complexity of these model-based techniques limits their applicability in low latency WNCS applications.

To address the complexity associated with model-based approaches, researchers have increasingly turned to machine learning techniques. \cite{dl_wncs_2024} presents a deep reinforcement learning (DRL) based algorithm that optimizes controllers and schedulers by leveraging both model-free and model-based data, incorporating awareness of sensor AoI states and dynamic channel conditions. Similarly, \cite{drl_wncs_2022} proposes a DRL-based co-design strategy for jointly optimizing WNCSs to achieve the best control performance despite network uncertainties like delay and variable throughput. Recently, \cite{hamida_amir_2024} introduces a DRL model grounded in optimization theory for co-designing control and communication systems, where the sampling period, blocklength, and packet error probability are optimized. However, ensuring feasibility and safety in DRL-based optimal resource allocation remains a challenge in WNCSs. While these studies leverage DRL for co-design, they do not incorporate mechanisms to explicitly enforce constraint satisfaction during training and decision-making. In contrast, our approach integrates optimization theory with safe DRL to guarantee constraint compliance throughout the learning process and final policy execution.

Numerous safe reinforcement learning (RL) methods have been explored across various domains. \cite{streaming_saferl_2021} investigates secure video streaming with low energy consumption in rotary-wing UAV-enabled wireless networks. The constraints satisfied with this safe method include secrecy timeout probability (STP) to meet long-term time delay requirements, while optimizing energy efficiency through video quality selection, power allocation, and trajectory optimization. A safe Deep Q-Network (DQN) is employed to develop a policy set using a Lyapunov function, dynamically satisfying the constraints of a Constrained Markov Decision Process (CMDP). However, this method ensures constraint satisfaction only on average over time using Lyapunov optimization, rather than guaranteeing that each decision individually meets all constraints, which may be insufficient for the strict real-time requirements of URLLC. On the other hand, \cite{safe_robot_2022} presents a safety layer for robot navigation, where a convex optimization problem is formulated to identify the safest action that adheres to constraints while minimizing deviation from the current action. This method is effective for convex constraint systems but has been applied
in domains such as robotic navigation, rather than communication networks. In \cite{nguyen2022fuzzy}, fuzzy logic is used to shape the reward function in Q-learning for vehicular crowdsensing. However, designing effective rules is non-trivial in systems with conflicting constraints, which limits their applicability to real-time URLLC scenarios, such as WNCSs. Moreover, \cite{enabling_robust_2022} introduces a teacher-student approach for resource allocation in networking, where the teacher employs simple white-box logic rules to guide the DRL student. While this method allows the teacher to offer precise action advice, the use of simple logic rules falls short in ensuring safety within RL-based URLLC systems. These rules lack the flexibility, adaptability, and scalability required to operate effectively in the complex, dynamic, and uncertain environments that characterize real-time communication networks. Although these safe RL methods have advanced the field, none has been specifically applied to wireless networked control systems (WNCSs) to ensure URLLC under real-time constraints.

This paper presents a safe DRL framework grounded in optimization theory for resource allocation in WNCSs. With the goal of minimizing overall power consumption, resource allocation seeks to identify the optimal values for key control and communication system parameters, including sampling period, blocklength, and packet error probability. The framework addresses constraints, including finite blocklength, PAoI violation probability, maximum transmit power, and schedulability, which are critical for the joint design of control and communication systems. To enhance efficiency, the proposed framework leverages the optimality conditions derived from mathematical system models, reducing the number of decision variables and, consequently, the amount of training data required for the DRL model. Additionally, it employs a teacher-student framework, wherein the teacher intervenes and recommends safe actions when the DRL agent (the student) fails to meet system requirements. The novel contributions of the paper are listed as follows:
\begin{itemize}
  \item We propose a safe DRL framework for the joint optimization of controller and communication systems involving multiple nodes, incorporating PAoI violation probability constraint while considering URLLC conditions within the finite blocklength regime, for the first time in the literature.
  \item We introduce a novel and direct formulation for the PAoI violation probability by integrating stochastic MATI and MAD constraints in a multi-node WNCS under URLLC conditions, utilizing finite blocklength theory, for the first time in the literature. 
  \item We propose a novel safe DRL approach grounded in optimization theory for optimal resource allocation in WNCSs, for the first time in the literature. The proposed approach reduces the training time for DRL by deriving optimality conditions and ensures system requirements are consistently met during exploration by employing safe DRL techniques. 
  \item We evaluate the performance of the proposed optimization theory based safe DRL approach through extensive simulations, comparing it against rule-based DRL and several optimization theory based DRL approaches across various network sizes and constraint parameters.
\end{itemize}

The rest of this paper is organized as follows. In Section \ref{sys_model}, the system model and associated assumptions are described. Section \ref{opt_wncs} presents the joint optimization problem of control and communication systems. The proposed optimization theory based safe DRL approach is provided in Section \ref{safe_drl}. Simulation results are presented in Section \ref{results}. Finally, concluding remarks are given in Section \ref{conclusion}.

\section {System Model and Assumptions} \label{sys_model}

 We consider a Wireless Networked Control System architecture comprising multiple plants, each continuously monitored by a central controller through a wireless network, as illustrated in Fig. \ref{wn_system}. The system consists of $\mathit{N}$ sensors, each of which periodically measures and transmits the corresponding plant state. The transceiver at each sensor node is assumed to have a single antenna. Sensor node $\mathit{i}$ transmits the plant state information in small packets of length $\mathit{L_i}$ by using $\mathit{m_i}$ symbols, referred to as the blocklength, where $\mathit{i} \in \{1,2,\dots, \mathit{N}\}$. Due to the inherent unreliability of the wireless communication channel, there is a possibility that the transmitted data may not reach the controller. For simplicity, packet errors are modeled as a Bernoulli random process. The controller utilizes the most recent state update to calculate and send new control commands to the actuators. The actuators establish and maintain stable communication links with the controller, ensuring the reliable and uninterrupted reception of control commands \cite{ncs_survey_2020}. We denote the sampling period, blocklength, packet transmission delay, and packet error probability of sensor node $\mathit{i}$ by \(\mathit{h_i}\), \(\mathit{m_i}\), \(\mathit{d_i}\), and \(\mathit{p_i}\), respectively, for $\mathit{i} \in \{1,2,\dots, \mathit{N}\}$. 

Time Division Multiple Access (TDMA) is assumed as the channel access method due to its guaranteed delay, making it suitable for industrial control applications \cite{ANSI38}. The channel time is segmented into frames, with each frame containing time slots. The first slot is reserved for the beacon frame. The controller periodically broadcasts this beacon to synchronize and update scheduling across the nodes in the WNCS. Nodes are assigned specific time slots for data transmission, with optimal settings for transmission power, blocklength, and sampling period tailored to each node's requirements. To ensure fixed determinism over the channel coherence time, during which wireless conditions remain stable, each sensor node is assigned a dedicated, fixed-duration time slot. These slots repeat periodically in alignment with the sampling interval, enabling predictable and consistent transmissions. Nodes do not transmit concurrently, and the network manager continuously monitors the packet error rate. Each node can operate in three modes: active (sensing, transmitting, and receiving), sleep, and transient. Energy consumption is considered only during the active mode, as sleep and transient modes consume minimal energy \cite{cui2005energy}.

Next, we discuss the optimization problem formulations related to the performance of the wireless communication and control systems. These include PAoI violation probability, power consumption, and schedulability expressions of the underlying communication system in the finite blocklength regime.

\begin{figure}[htbp]
\centerline{\includegraphics[width = .5\textwidth]{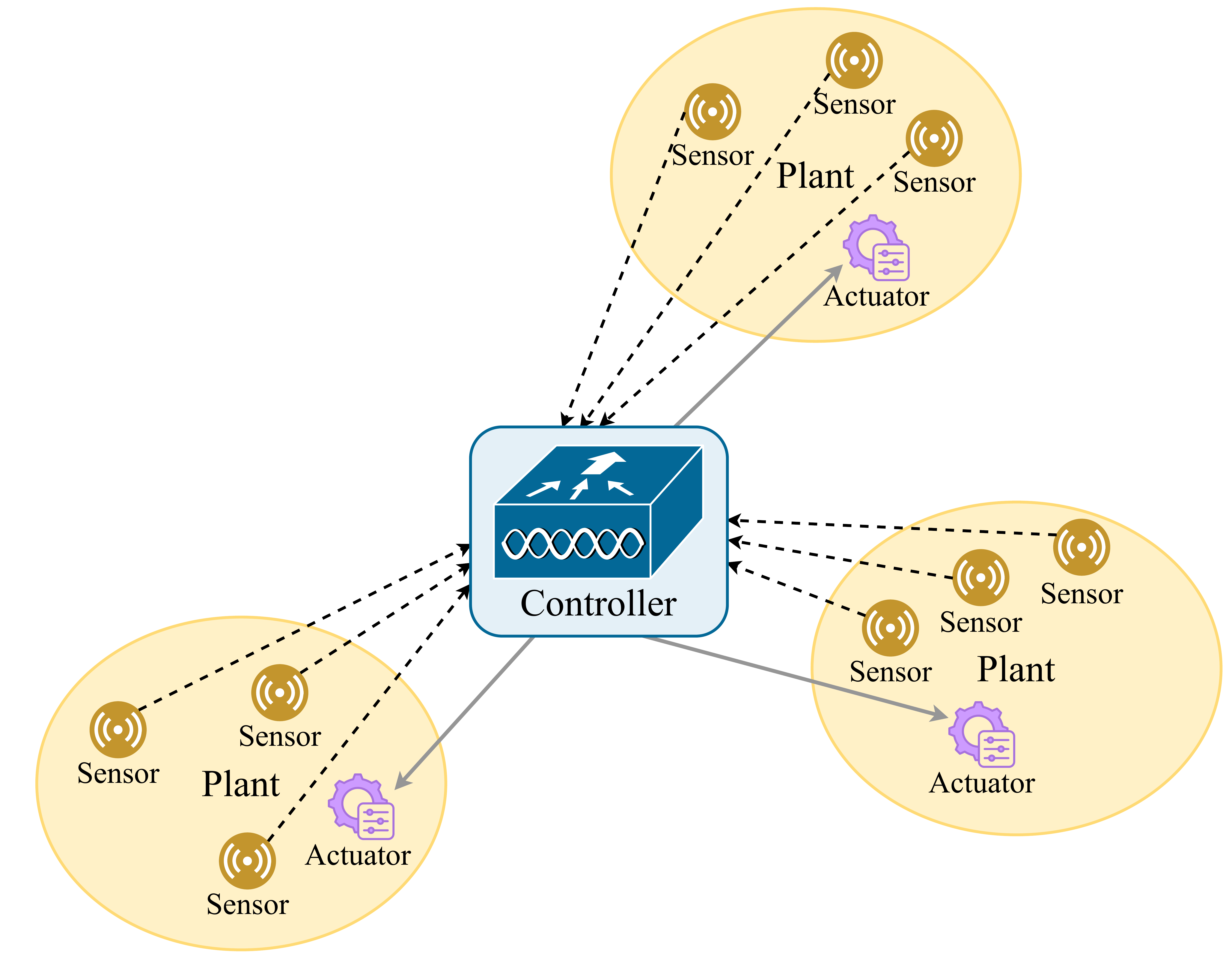}}
\caption{Overview of the WNCS setup.}
\label{wn_system}
\end{figure}

\vspace{-0.2cm}

\subsection{Peak AoI Violation Probability}

AoI is a critical metric that measures the freshness of data in a system by tracking the time elapsed since the last update was received. PAoI represents the maximum AoI just before a new update is received, indicating the worst-case scenario in terms of data freshness. In the considered system, PAoI is a key performance indicator (KPI), since maintaining up-to-date information is vital for timely and accurate decision-making, directly influencing the quality of control decisions and overall system stability. The PAoI violation probability quantifies the likelihood that the PAoI exceeds a predefined critical threshold. Keeping PAoI violation probability under a threshold value ensures that the system maintains real-time data freshness. This constraint is formulated as
\begin{equation}
    P[ \mathit{A_i}(\mathit{h_i}, \mathit{d_i}(\mathit{m_i}), \mathit{p_i}) \geq \mathit{\alpha} ] \leq (1-\mathit{\delta}),
    \label{paoi_violation}
\end{equation}
where $\mathit{A_i}$ denotes the PAoI of sensor node $\mathit{i}$ and is a function of the sampling period $\mathit{h_i}$, packet transmission delay $\mathit{d_i}$, and packet error probability $\mathit{p_i}$; the delay $\mathit{d_i}$ is a function of the blocklength $\mathit{m_i}$; $\mathit{\alpha}$ denotes the endurable peak AoI and $\textstyle (1-\mathit{\delta})$ indicates the maximum tolerated probability for peak AoI violation. The value of $\mathit{\alpha}$ is application-specific and can vary significantly. For some applications, particularly those requiring ultra-low latency, $\mathit{\alpha}$ might be on the order of a few milliseconds, while for others, such as periodic monitoring or less time-sensitive applications, it could extend to several minutes
\cite{paoi_constraint_2023,min_paoi_2021}.

The PAoI violation probability can be rewritten as
\begin{equation}
    P[ \mathit{\mu_i}(\mathit{h_i}, \mathit{d_i}(\mathit{m_i}), \mathit{p_i}) + \mathit{d_i}(\mathit{m_i}) \geq \mathit{\alpha} ] \leq (1-\mathit{\delta}),
    \label{mati_mad_violation}
\end{equation}
where $\mathit{\mu_i}$ denotes the time between two subsequent state vector updates of sensor node $\mathit{i}$. This expression combines MAD and stochastic MATI constraints in \cite{hamida_amir_2024}. For given values of $\mathit{h_i}$ and $\mathit{\alpha}$, there are $\textstyle \left \lfloor  \frac{\mathit{\alpha} - \mathit{d_i}(\mathit{m_i})}{\mathit{h_i}} \right \rfloor$ opportunities of the state vector update reception. We model the packet error $\mathit{p_i}$ as a Bernoulli random process and rewrite the PAoI violation probability constraint as \cite{hamida_amir_2024}
 
\begin{equation}
    \mathit{p_i}^{\left \lfloor  \frac{\mathit{\alpha} - \frac{\mathit{m_i}}{\mathit{B}}}{\mathit{h_i}} \right \rfloor} \leq (1-\mathit{\delta}),
    \label{mati_delay_violation_final}
\end{equation}
where $\mathit{B}$ is the bandwidth of the wireless channel and $\textstyle \mathit{d_i}(\mathit{m_i}) = \frac{\mathit{m_i}}{\mathit{B}}$ accounts for the packet transmission delay based solely on the blocklength $\mathit{m_i}$. Given that fixed determinism is assumed, the queuing delay is excluded from the analysis. 

\subsection{Power Consumption}
Power consumption is a vital consideration in energy-constrained systems with limited energy sources. Efficient power management is essential to maximize the operational lifespan of these systems, ensuring reliable performance over extended periods without the need for frequent recharging or battery replacement. Average power consumption $\mathit{W_i}$ of node $\mathit{i}$ is a function of the sampling period $\mathit{h_i}$, packet error probability $\mathit{p_i}$, and delay $\mathit{d_i}(\mathit{m_i})$ such that:
\begin{equation}
    \mathit{W_i}(\mathit{h_i},\mathit{d_i}(\mathit{m_i}),\mathit{p_i}) = \frac{(\mathit{W_{\text{tx},i}} + \mathit{W_i^c})\mathit{d_i}(\mathit{m_i})}{\mathit{h_i}},
    \label{power}
\end{equation}
where $\mathit{W_{\text{tx},i}}$ is the transmit power of node $\mathit{i}$ and $\mathit{W_i^c}$ is the circuit power consumption when the transmitter is in the active mode. 

The coding rate, also referred to as the data rate of a communication system, is the ratio of the number of information bits to the total number of channel uses (also known as the blocklength). In URLLC scenarios, where the blocklength per frame is small, the decoding error probability becomes significant and cannot be neglected. The coding rate $\mathit{R_i}$ (in bits/sec/Hz) for node $\mathit{i}$ under the finite blocklength regime can be approximated as \cite{hamida_amir_2024}
\begin{equation}
   \mathit{R_i} \approx \log_2(1+\mathit{\gamma_i}) - \sqrt{\frac{\mathit{V_i}}{\mathit{m_i}}} \frac{Q^{-1}(\mathit{\epsilon_i})}{\ln{(2)}},
   \label{coding_rate}
\end{equation}
where $\textstyle \mathit{\gamma_i}= \frac{\mathit{W_{\text{tx},i}} \lvert \mathit{g_i} \rvert }{\mathit{\sigma^2}}$ is the signal-to-noise ratio (SNR) at the controller; $\mathit{g_i}$ is the channel gain from node $\mathit{i}$ to the controller; $\mathit{\sigma^2}$ is the noise power; $ \textstyle \mathit{V_i}=1-(1+\mathit{\gamma_i})^{-2}$ is the channel dispersion; $\mathit{\epsilon_i}$ is the decoding error probability, and $Q^{-1}$ denotes the inverse of the Gaussian $Q$ function. 

In a short frame structure, the frame duration is smaller than the channel coherence time, resulting in a quasi-static channel. Therefore, all symbols within a packet experience nearly identical fading conditions. Considering that decoding errors are the primary cause of packet errors, the decoding error probability can be equated with the packet error probability \cite{durisi2016toward}. As a result, throughout the remainder of this paper, we can safely replace the decoding error probability $\mathit{\epsilon_i}$ with the packet error probability $\mathit{p_i}$ in our equations.

The transmit power $\textstyle \mathit{W_{\text{tx},i}}$ of node $\mathit{i}$ can be extracted from Eq. (\ref{coding_rate}) by substituting $\mathit{R_i}$ with $\textstyle \frac{\mathit{L_i}}{\mathit{m_i}}$ as \cite{hamida_amir_2024} 
\begin{equation}
    \mathit{W_{\text{tx},i}} = \mathit{C_{i1}} \left[\exp\left( \frac{Q^{-1} (\mathit{p_i})}{\sqrt{\mathit{m_i}}} + \frac{\ln(2)\mathit{L_i}}{\mathit{m_i}}\right) - 1\right],
\end{equation}
where $\textstyle \mathit{C_{i1}} = \frac{\mathit{\sigma^2}}{\mathit{|g_i|}}$ and the channel dispersion is approximated as $\mathit{V_i} \approx 1, \forall \mathit{i} \in \mathcal {N}$. While the unit dispersion assumption holds exactly in the high SNR regime, it is also commonly employed in the medium SNR range under practical blocklengths to simplify analysis. This approximation enables tractable problem formulation without significantly compromising accuracy, as supported in prior studies \cite{ren2019joint,snr_approx_2023,channel_fbl_2014}.  By substituting $\mathit{d_i}(\mathit{m_i})$ with $\textstyle \frac{\mathit{m_i}}{\mathit{B}}$, Eq. (\ref{power}) can be rewritten as
\begin{equation}
\begin{split}
        \mathit{W_i}&(\mathit{h_i},\mathit{d_i}(\mathit{m_i}),\mathit{p_i}) =  \\
        & \frac{ \mathit{C_{i1}} \mathit{m_i} \mathit{}}{\mathit{h_i B }} \left[\exp\left( \frac{Q^{-1} (\mathit{p_i})}{\sqrt{\mathit{m_i}}}  +  \frac{\ln(2)\mathit{L_i}}{\mathit{m_i}}  \right) - 1\right] + \frac{\mathit{W_i^c m_i} }{\mathit{h_i B}}.  
\end{split}
\label{total_power}
\end{equation}

\subsection{Schedulability}
The schedulability constraint manages the assignment of transmission times to multiple sensor nodes in a network, taking into account wireless transmission limitations. We assume that no two nodes can transmit simultaneously. The transmissions are scheduled using the Earliest Deadline First (EDF) algorithm, a dynamic scheduling method that prioritizes transmissions based on their proximity to deadlines \cite{Dertouzos46}. We define the schedulability constraint as proposed by \cite{sadi2014minimum}

\begin{equation}
     \sum_{\mathit{i}=1}^{\mathit{N}} \frac{\mathit{d_i}(\mathit{m_i})}{\mathit{h_i}} \leq \mathit{\beta}, 
     \label{schedule_constraint}
\end{equation}
where $\mathit{\beta}$ is the utilization bound, constrained by $\textstyle 0 < \mathit{\beta} \leq 1$. Each node $\mathit{i}$ is allocated a transmission fraction of $\textstyle \frac{\mathit{d_i}(\mathit{m_i})}{\mathit{h_i}}$. Since nodes cannot transmit simultaneously, the sum of these fractions must fit within the total schedule length. Therefore, $\mathit{\beta}$ cannot exceed one.


\section{Joint Optimization of Control and Communication Systems}  \label{opt_wncs}
In this section, we formulate the joint optimization problem for ultra-reliable WNCSs in the finite blocklength regime. The objective is to minimize the power consumption of the communication system while adhering to constraints on maximum blocklength, peak AoI violation probability, maximum transmit power, and schedulability. The problem is formulated as follows:

\begin{subequations}
\begin{align}
    \underset{ \mathit{h_i},\mathit{m_i},\mathit{p_i}  \mathit{i} \in [1,N]}{min} \quad \sum_{\mathit{i}=1}^{\mathit{N}}  \mathit{W_i}(\mathit{h_i},\mathit{d_i}(\mathit{m_i}),\mathit{p_i}) \label{op_objective}
\end{align}
\quad \quad \text{s.t.} 
\begin{equation}
 \quad  \mathit{m_i} \leq \mathit{M_{\text{th}}}, \quad \forall \mathit{i} \in [1,\mathit{N}] 
\label{op_max_blocklength_cons}
\end{equation}
\begin{equation}
 {\left \lfloor  \frac{\mathit{\alpha} - \frac{\mathit{m_i}}{\mathit{B}}}{\mathit{h_i}} \right \rfloor} \ln{\mathit{p_i}} - \ln(1-\mathit{\delta})
\leq 0, \quad \forall \mathit{i} \in [1,\mathit{N}]  \\ \\
\label{op_paoi_violate_cons}
\end{equation}
\begin{equation}
  0 < \mathit{h_i} \leq \left(\mathit{\alpha} - \frac{\mathit{m_i}}{\mathit{B}}\right),\quad \forall \mathit{i} \in [1,\mathit{N}] \\ 
\label{op_sampling_cons}
\end{equation}
\begin{equation}
  0 < \mathit{p_i} < 1, \quad \forall \mathit{i} \in [1,\mathit{N}] \\ 
\label{op_error_prob_cons}
\end{equation}
\begin{equation}
\begin{split}
    \mathit{W_{\text{tx},i}} \leq \mathit{W_{\text{tx}}^{\text{max}}}, \quad \forall \mathit{i} \in [1,\mathit{N}] 
\end{split}
\label{op_max_wt_cons}
\end{equation}
\begin{equation}
    \sum_{\mathit{i}=1}^{\mathit{N}} \frac{\mathit{d_i}(\mathit{m_i})}{\mathit{h_i}} \leq \mathit{\beta}.
\label{op_sched_cons}
\end{equation}
\label{optimization_problem}
\end{subequations} 

The decision variables of the optimization problem are the sampling period $\mathit{h_i}$, packet error probability $\mathit{p_i}$, and blocklength $\mathit{m_i}, \mathit{i} \in [1,\mathit{N}]$. Eq. (\ref{op_objective}) gives the objective of the problem for the minimization of the total power consumption, as derived in Eq. (\ref{total_power}). Eq. (\ref{op_max_blocklength_cons}) represents the maximum blocklength limit, which specifies the maximum number of symbols or channel uses allowed for packet transmission. Eq. (\ref{op_paoi_violate_cons}) represents the PAoI violation probability constraint. Additionally, Eq. (\ref{op_sampling_cons}) states the constraints on sampling periods required to ensure the freshness of data updates. Eq. (\ref{op_error_prob_cons}) provides the limits on error probability. Eq. (\ref{op_max_wt_cons}) is the maximum transmit power constraint to prevent excessive energy consumption, minimize interference with other devices, and comply with regulatory limits. Finally, Eq. (\ref{op_sched_cons}) represents the schedulability constraint.

This optimization problem is a non-convex Mixed-Integer Programming problem; thus, any solution for global optimum is difficult \cite{Convex47}. 

\section {Optimization Theory Based Safe DRL} \label{safe_drl}
Optimization theory-based safe DRL combines domain-specific optimization theory formulation, data-driven deep learning techniques, and a teacher-student architecture that guarantees a secure exploration. This hybrid approach minimizes the amount of training data required in model-agnostic DRL while safely exploring the state space. The approach is divided into two stages: 1) Optimization Theory Stage: Focuses on simplifying the optimization problem using optimality conditions and domain knowledge from communication systems. 2) Safe DRL stage: Involves training the deep learning model with insights from student experiences and optimized teacher advice. The details of these stages are provided in Sections \ref{sec:optimization-theory-stage} and \ref{sec:drl-stage}. 

\subsection{Optimization Theory Stage}\label{sec:optimization-theory-stage}
In the optimization theory stage, we apply optimality conditions for the sampling period, packet error probability, and the blocklength, following the approach outlined in \cite{hamida_amir_2024}. This analysis allows us to focus on the blocklength, $m_i$, and determine the other variables, i.e., sampling period and packet error probability, by using these conditions. 

In Lemma 1, we relate the optimal sampling period and packet error probability to the PAoI violation probability constraint.

\begin{lemma}
Let $\mathit{h_i^*}$, $\mathit{p_i^*}$, and $\mathit{m_i^*}$ denote the optimal values of the sampling period, packet error probability, and blocklength, respectively. The relationship between the optimal sampling period and packet error probability is expressed as
\begin{equation}
    \frac{\mathit{\alpha} - \frac{\mathit{m_i}}{\mathit{B}}}{\mathit{h_i^*}} = \frac{\ln(1-\mathit{\delta})}{\ln{\mathit{p_i^*}}} = \mathit{k_i^*},
\label{opt_ki}
\end{equation}
where $\mathit{k_i^*}$ is a positive integer. 
\end{lemma}

\begin{proof} 
We prove the Lemma by contradiction, following the methods in  \cite{hamida_amir_2024} and \cite{sadi2014minimum}. We start by assuming that $\textstyle \frac{\mathit{\alpha} - \frac{\mathit{m_i}}{\mathit{B}}}{\mathit{h_i^*}}$ is not a positive integer such that $ \textstyle \left \lfloor \frac{\mathit{\alpha} - \frac{\mathit{m_i}}{\mathit{B}}}{\mathit{h_i^*}} \right \rfloor < \frac{\mathit{\alpha} - \frac{\mathit{m_i}}{\mathit{B}}}{\mathit{h_i^*}}$. If we increase $\mathit{h_i^*}$ to make $\textstyle \left \lfloor \frac{\mathit{\alpha} - \frac{\mathit{m_i}}{\mathit{B}}}{\mathit{h_i^*}} \right \rfloor = \frac{\mathit{\alpha} - \frac{\mathit{m_i}}{\mathit{B}}}{\mathit{h_i^*}}$, while ensuring the upper bound in Eq. (\ref{op_sampling_cons}) is not violated, the constraint in Eq. (\ref{op_paoi_violate_cons}) remains satisfied, as $ \textstyle \left \lfloor \frac{\mathit{\alpha} - \frac{\mathit{m_i}}{\mathit{B}}}{\mathit{h_i^*}} \right \rfloor$ does not change, and the constraint (\ref{op_sched_cons}) is still met. However, the objective function in Eq. (\ref{op_objective}) decreases since it is a monotonically decreasing function of $\mathit{h_i}$. Thus, keeping the inequality would contradict the optimality of $\mathit{h_i^*}$, proving that the equality must hold.

Similarly, we can prove $\textstyle \frac{\mathit{\alpha} - \frac{\mathit{m_i}}{\mathit{B}}}{\mathit{h_i^*}} = \frac{\ln(1-\mathit{\delta})}{\ln{\mathit{p_i^*}}}$ by contradiction. Suppose $\textstyle \frac{\mathit{\alpha} - \frac{\mathit{m_i}}{\mathit{B}}}{\mathit{h_i^*}} > \frac{\ln(1-\mathit{\delta})}{\ln{\mathit{p_i^*}}}$. By increasing $\mathit{p_i^*}$ such that $\textstyle \frac{\mathit{\alpha} - \frac{\mathit{m_i}}{\mathit{B}}}{\mathit{h_i^*}} = \frac{\ln(1-\mathit{\delta})}{\ln{\mathit{p_i^*}}}$, the constraint in Eq. (\ref{op_max_wt_cons}) still holds since the power consumption of node $\mathit{i}$ is a non-increasing function of $\mathit{p_i}$. Nonetheless, the power consumption function in Eq. (\ref{op_objective}) decreases since it is a monotonically decreasing function of $\mathit{p_i}$. Since reducing power consumption is the optimization goal, allowing this inequality would imply a suboptimal choice of $\mathit{p_i^*}$, confirming that the equality must hold in the optimal solution.
\end{proof}

\begin{lemma}
The optimal value of $\mathit{k_i}$ in Lemma 1, denoted as $\mathit{k_i^*}$, is expressed as a function of $\mathit{m_i}$ and given by
\begin{equation}
\begin{split}
    \mathit{k_i^*} &= \max \\
    &\left[1, \left \lfloor \quad  \frac{\ln(1-\mathit{\delta})} {\ln \left[Q \left [\sqrt{\mathit{m_i}} \ln \left(\frac{\mathit{W_{\text{tx}}^{\text{max}}} }{\mathit{m_i C_{i1}}} + 1 \right) - \frac{\ln{(2)}\mathit{L_i}}{\sqrt{\mathit{m_i}}} \right]  \right] } \right \rfloor \right].
\end{split}
\label{ki_formula}
\end{equation}
\end{lemma}

\begin{proof}
The power consumption in Eq. (\ref{op_objective}) increases monotonically with $\mathit{k_i^*}$. Therefore, $\mathit{k_i^*}$ must be the smallest positive integer that satisfies the constraints (\ref{op_max_blocklength_cons})-(\ref{op_sched_cons}). In Eq. (\ref{ki_formula}), the second term of the maximum function represents the minimum value of $\mathit{k_i}$ obtained in Lemma 1.  
\end{proof}

By using (\ref{opt_ki}), the transmit power can be rewritten as 
\begin{equation}
  \mathit{W_{\text{tx},i}^*} = \mathit{C_{i1}} \left[\exp\left({ \frac{Q^{-1} ((1-\mathit{\delta})^{1/\mathit{k_i^*}})}{\sqrt{\mathit{m_i}}} + \frac{\ln(2)\mathit{L_i}}{\mathit{m_i}}} \right) - 1\right],  
\end{equation} 
while the schedulability constraint can be reformulated as 
  \begin{equation}
     \sum_{\mathit{i}=1}^{\mathit{N}} \frac{\mathit{m_i  k_i^*}}{\mathit{B\alpha} - \mathit{m_i}} \leq \mathit{\beta}.
\end{equation} 
Also, Eq. (\ref{total_power}) is updated as 
  \begin{equation}
     \begin{split}
        \mathit{W_i^*}(\mathit{m_i}) &= \frac{\mathit{C_{i1} m_i k_i^*}}{\mathit{B\alpha} - \mathit{m_i}}  \left[ \exp\left({ \frac{Q^{-1} {(1-\mathit{\delta})} ^\frac{1}{\mathit{k_i^*}} }{\sqrt{\mathit{m_i}}} + \frac{\ln(2)\mathit{L_i}}{\mathit{m_i}} }\right) - 1\right] \\ 
        &+ \frac{ \mathit{W_i^c m_i k_i^*}}{\mathit{B\alpha} - \mathit{m_i}}.   
\end{split}
\end{equation} 
The joint optimization problem (\ref{optimization_problem}) can then be reformulated as

\begin{subequations} 
\begin{align}
    \underset{{\mathit{m_i}}, i \in [1,\mathit{N}]}{\min} \quad \sum_{\mathit{i}=1}^{\mathit{N}} \mathit{W_i^*(m_i)}\label{eq:modified_goal} 
\end{align}
\quad \text{s.t.} 
\begin{equation}
     \mathit{m_i} \leq \mathit{M_{\text{th}}}, \quad \forall \mathit{i} \in [1,\mathit{N}] \\
\label{eq:modified_m_th}
\end{equation}
\begin{equation}
  \mathit{W_{\text{tx},i}^*} \leq \mathit{W_{\text{tx}}^{\text{max}}}, \quad \forall \mathit{i} \in [1,\mathit{N}]  \\
\label{modified_op_max_wt_cons}
\end{equation}
\begin{equation}
    \sum_{\mathit{i}=1}^{\mathit{N}} \frac{ \mathit{m_i  k_i^*}}{\mathit{B\alpha} - \mathit{m_i}} \leq \mathit{\beta}.
\label{modified_op_sched_cons}
\end{equation}
\label{modified_optimization_problem}
\end{subequations} 

\subsection{Safe DRL Stage}\label{sec:drl-stage}
Now that the problem is simplified, we propose a safe deep reinforcement learning based solution to solve the non-convex optimization problem. After solving the simplified version, the optimal packet error probability and sampling period can be obtained by utilizing the optimality conditions obtained in Section \ref{sec:optimization-theory-stage}. 

We now provide an overview of safe DRL, the teacher-student control mechanism, and its integration with the DRL model.

\begin{figure*}[t]
\centering
\includegraphics[width = .95\textwidth]{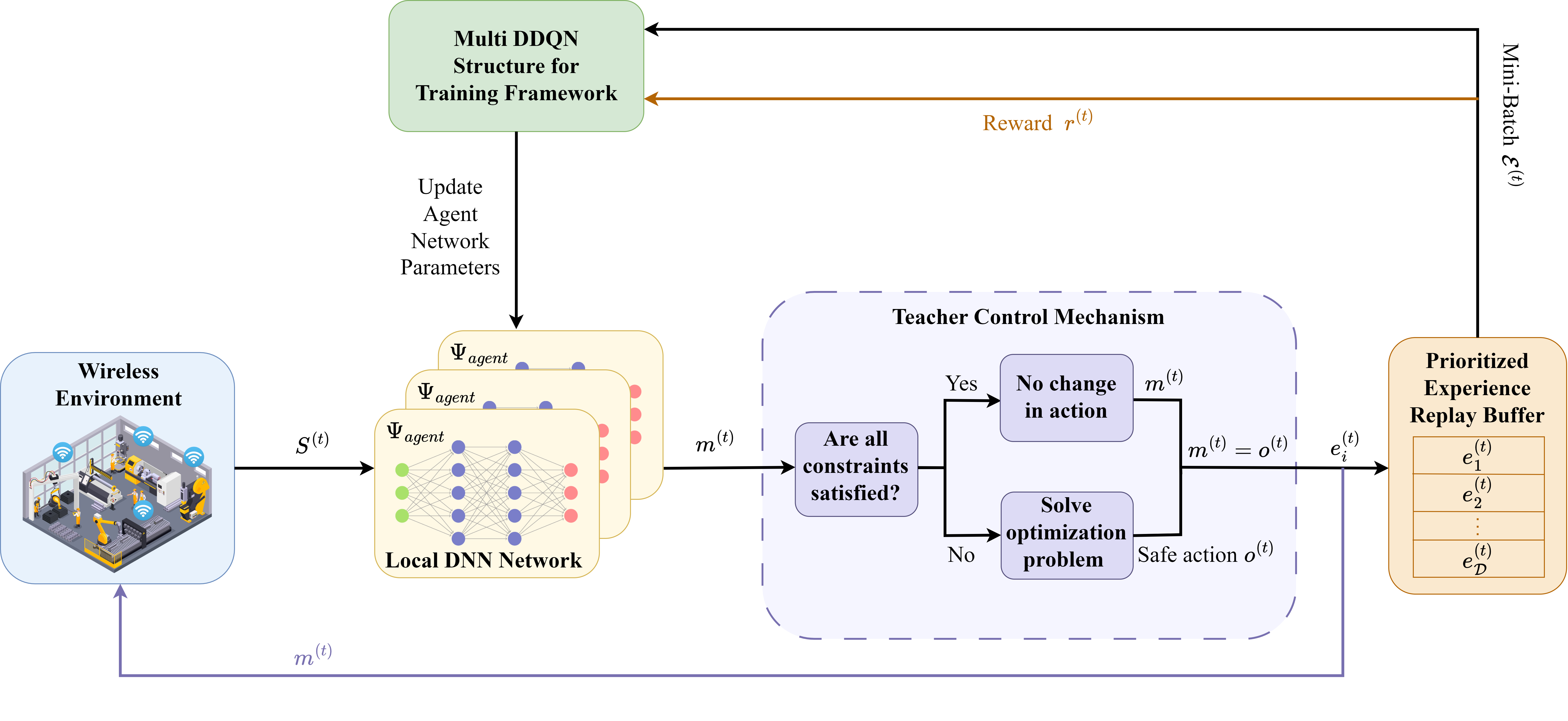}
\caption{Proposed safe DRL algorithm.}
\label{system_diagram}
\end{figure*}

\subsubsection{Overview of Safe DRL} \label{subsec:overview-drl}

RL is a machine learning approach where an agent learns to make decisions by interacting with an environment, aiming to maximize cumulative rewards over time. The agent selects actions based on the current state of the environment, and these actions affect both immediate rewards and future states. RL methods, such as Q-learning and DQN, allow the agent to learn optimal strategies through trial and error, which often involves exploring suboptimal or unsafe actions during the learning phase. In many practical applications, such as autonomous driving or wireless communications, it is crucial to prevent unsafe behaviors, especially during the learning process.

Safe RL addresses this challenge by ensuring that the agent operates within predefined safety constraints. Various approaches are used to enforce these constraints, such as modifying the reward function to penalize constraint violations, extending Markov Decision Processes (MDPs) to include safety constraints, and using teacher-student frameworks where the teacher oversees the agent’s actions, ensuring that unsafe actions are corrected in real-time. In other cases, Lyapunov-based methods are employed to guarantee stability and safety throughout the learning process. Safe RL is critical in scenarios where violating constraints can lead to costly or dangerous outcomes, ensuring that the agent can still explore and learn effectively without compromising safety.

In the context of URLLC systems, safe RL plays a vital role due to the strict requirements on latency, reliability, and power consumption. Traditional RL methods, which lack built-in safety mechanisms, may result in the exploration of unsafe actions, potentially violating these stringent constraints. Safe RL, however, ensures that the agent can optimize key performance metrics, such as minimizing latency or power consumption, while consistently adhering to reliability and latency thresholds. This ability to balance optimization with guaranteed adherence to system constraints makes safe RL indispensable for meeting the real-time, mission-critical demands of URLLC applications, where even minor deviations can have significant impacts on service quality and system performance.

\subsubsection{Teacher-Student Control Mechanism} \label{subsec:proposed-drl}

\paragraph{Teacher Role in Ensuring Safety} \label{subsubsec:teacher-ensure}
The teacher in the teacher-student framework plays a critical role in ensuring that the agent's actions remain within predefined safety constraints. The teacher continuously monitors the student's decisions and provides real-time guidance to prevent constraint violations. The following are key stages where constraint violations may occur, along with the mechanisms used by the teacher to mitigate risks and why these interventions are essential:

\begin{itemize}
\item \textit{Random Experience Accumulation Phase:}
In the early phase of random experience accumulation, the agent takes exploratory actions without any prior knowledge, significantly increasing the risk of constraint violations. Since these actions are chosen randomly, they may breach critical limits on latency or packet error rates. To prevent this, the teacher continuously monitors the agent’s actions and intervenes when necessary by adjusting them to the closest feasible option that adheres to all predefined safety constraints. This ensures that the experience stored in the replay buffer is not only valuable for effective learning and future decision-making but also safe.

\item \textit{Exploration in Training Phase:}
During training, the agent employs an epsilon-greedy policy to balance exploration and exploitation. While this strategy facilitates effective learning, the exploration phase (triggered with probability epsilon) can still result in the selection of unsafe actions that violate key constraints. To mitigate this, the teacher carefully evaluates each action proposed under the epsilon-greedy policy and intervenes when needed, adjusting any unsafe choices to ensure they remain within the predefined safety boundaries. This guidance helps the agent maintain compliance with constraints, promoting safer exploration throughout the learning process.

\item \textit{Exploitation in Training Phase:}
The agent's decision-making, guided by learned Q-values, can occasionally recommend actions that fail to meet safety requirements, often due to incomplete learning or suboptimal policy updates. To prevent this, the teacher continuously evaluates the agent's suggested actions, intervening whenever a potential violation occurs. This ensures that the learning process does not reinforce unsafe behavior, maintaining a stable and reliable learning environment.

\item \textit{Testing Phase:}
During the testing phase, the agent executes actions without further training or exploration, relying on the learned policy to perform optimally. However, if the policy has learned suboptimal strategies that occasionally violate constraints, the teacher's role is to monitor and correct these actions in real-time. By doing so, the teacher guarantees that the agent's performance remains within safety boundaries, even in the absence of explicit penalties for violations in the reward function.
\end{itemize}

By intervening at these critical points, the teacher ensures that the agent's actions consistently adhere to the required safety constraints, preventing any violations that could compromise system performance or safety. This constant guidance not only helps maintain a safe learning environment but also accelerates the convergence of the agent's policy towards optimal, safe solutions.

\paragraph{State, Action, and Reward Design} \label{subsubsec:key-elements}
The key components of the teacher controlled DRL are outlined below.

\begin{itemize}
\item \textit{States:}
The state space for each agent in the system represents factors that directly influence the environment, with each agent’s next state being affected by the actions of all agents collectively. The state for agent $\mathit{i}$ at time $\mathit{t}$ is defined as 
\begin{equation} \label{eq:8}
    \mathit{s_i^{(t)}} = \left(\mathbf{\mathit{m}}^{(t-1)}, \mathit{R_i^{(t-1)}}, \mathbf{\mathit{W}}^{(t-1)}, \mathcal{\mathit{P}}^{(t-1)}, \mathit{\gamma_i^{(t)}}, \mathbf{\mathit{g}}^{(t)}\right).
\end{equation}
In this state space, $\mathbf{\mathit{m}}^{(t-1)}$ is a vector containing the blocklengths selected by each node at the previous time step, with the $\mathit{i}$th element specifically representing the blocklength chosen by node $\mathit{i}$. The variable $\textstyle \mathit{R_i^{(t-1)}}$ denotes the data rate of node $\mathit{i}$ from the prior time step. The vector $\mathbf{\mathit{W}}^{(t-1)}$ includes the power levels used by each node at the previous time step, where the $\mathit{i}$th element indicates the transmit power for node $\mathit{i}$. Additionally, $\textstyle \mathcal{\mathit{P}}^{(t-1)}$ captures the total power consumption of all nodes during the previous time step. The SNR for node $\mathit{i}$ at the current time step, $\textstyle \mathit{\gamma_i^{(t)}} = \frac{\mathit{g_{i}^{(t)} W_{\text{tx}, i}^{(t-1)}}}{\sigma^2}$, is calculated using the transmit power from the previous step and the current channel gain $\textstyle \mathit{g_i^{(t)}}$. Finally, the state space includes $\textstyle \mathbf{\mathit{g}}^{(t)}$, which is a vector of the current channel gains for all nodes. This comprehensive state representation combines historical performance data with current channel conditions, enabling agents to make well-informed decisions. An MDP is created using the parameters from the previous time step, further constraining the model for the blocklength. By incorporating both SNR and channel gains, the framework dynamically adapts to real-time variations in the wireless environment.

\item \textit{Actions:}
The agent $\mathit{i}$ selects the blocklength $\mathit{m_i^{(t)}}$ for transmission. The action spaces of all agents are identical, and that is given by
\begin{equation} \label{7_2}
    \mathit{A} = \{1, 2, ..., \mathit{M_{\text{th}}}\},
\end{equation}
which is decided based on the maximum blocklength constraint. 

\item \textit{Reward:}
The reward function is a crucial element in the RL process. Each agent aims to take optimal action to maximize the total reward. The reward corresponds to the objective function of the optimization problem (\ref{modified_optimization_problem}) and is given as follows: 
\begin{equation} \label{eq:9}
\begin{split}
    \mathit{r^{(t)}} = -\sum_{\mathit{i}=1}^{\mathit{N}}\mathit{W_i(m_i^{(t)})}.
\end{split}
\end{equation} 
In our DRL model, the reward function does not include penalties for constraint violations because we employ a safe method that inherently prevents any violations. As a result, the agent only explores safe actions, allowing the reward function to focus solely on optimizing the primary objective.

\end{itemize} 

\paragraph{Action Advising Strategy} \label{subsubsec:action-advice}

In its role as the control mechanism, the teacher identifies the closest action to the student's proposed action while satisfying all of the constraints similar to \cite{safe_robot_2022}. To achieve this, the teacher solves the following convex optimization problem:

\begin{subequations} 
\begin{equation}
    \underset{\mathit{o^{(t)}}}{\min} \quad \lVert \mathit{o^{(t)}} - \mathit{m^{(t)}} \rVert 
\label{eq:teacher_goal}
\end{equation}
\quad \text{s.t.} 
\begin{equation}
     \mathit{o_i^{(t)}} \leq \mathit{M_{\text{th}}}, \quad \forall \mathit{i} \in [1,\mathit{N}] \\
\label{eq:teacher_m_th}
\end{equation}
\begin{equation}
\resizebox{0.95\hsize}{!}{$%
    \mathit{C_{i1}} \left[\exp\left( \frac{Q^{-1} ((1-\mathit{\delta})^{1/\mathit{k_i^{(t)}}})}{\sqrt{\mathit{o_i^{(t)}}}} + \frac{\ln(2)\mathit{L_i}}{\mathit{o_i^{(t)}}}\right) - 1\right] \leq \mathit{W_{\text{tx}}^{\text{max}}}, 
    $%
    }%
\label{teacher_op_max_wt_cons}
\end{equation}
\begin{equation}
    \sum_{\mathit{i}=1}^{\mathit{N}} \frac{ \mathit{o_i^{(t)}  k_i^{(t)}}}{\mathit{B\alpha} - \mathit{o_i^{(t)}}} \leq \mathit{\beta},
\label{teacher_op_sched_cons}
\end{equation}
\label{teacher_optimization_problem}
\end{subequations} 
where $\mathit{o^{(t)}}$ is the current action advice vector from the teacher, with the $\mathit{i}$th element $\mathit{o_i^{(t)}}$ corresponding to the advised action for sensor node $\mathit{i}$; $\mathit{m^{(t)}}$ is the current action vector derived from the student's DRL algorithm, where the $\mathit{i}$th element $\mathit{m_i^{(t)}}$ reflects the student's action for sensor node $\mathit{i}$. The objective of the optimization problem (\ref{eq:teacher_goal}) is to minimize the Euclidean distance between the teacher's action advice vector and the student's action vector. The constraints (\ref{eq:teacher_m_th}), (\ref{teacher_op_max_wt_cons}) and (\ref{teacher_op_sched_cons}) correspond to the constraints (\ref{eq:modified_m_th}), (\ref{modified_op_max_wt_cons}) and (\ref{modified_op_sched_cons}) in the original problem, respectively. 

This optimization problem is solved using convex optimization solvers. The resulting action advice, $\mathit{o^{(t)}}$, is implemented by all nodes instead of the student's proposed action, thereby ensuring that the system operates safely without violating any requirements. While the teacher restricts unsafe exploration, it does not limit the discovery of better policies within the feasible space, allowing the DRL agent to optimize performance while maintaining constraint compliance.

\subsubsection{Integration with DRL Model} \label{subsec:drl-fw}

In the proposed safe DRL algorithm shown in Fig. \ref{system_diagram}, each sensor node functions as an independent agent within a multi-agent DRL architecture, where the teacher advice method is employed to ensure that system constraints are consistently satisfied. The evolution of the environment's states in this framework is determined by the collective actions of all agents. While multi-agent reinforcement learning models have demonstrated strong empirical performance, they often lack theoretical guarantees of convergence, as noted in \cite{nguyen2020deep}. To enhance both implementation efficiency and stability, a centralized training and execution approach is adopted, wherein all agents operate synchronously and select their actions simultaneously \cite{gupta2017cooperative}. 

In this setup, each agent $\mathit{i}$ observes the state $\mathit{s_i^{(t)}}$ and chooses an action $\mathit{m_i^{(t)}}$ via its respective local neural network. If the action of the agent, identified as the student, satisfies all the constraints, the teacher does not interfere. However, if any constraints are violated, the teacher intervenes by suggesting the closest permissible action $\mathit{o_i^{(t)}}$ that satisfies all the constraints. Then, the student action $\mathit{m_i^{(t)}}$ is set to $\mathit{o_i^{(t)}}$ to ensure compliance with the requirements. Upon executing the action, each agent receives a reward $\mathit{r_i^{(t)}}$, and the cumulative reward function $\mathit{r^{(t)}}$ across all agents is computed. At each time step $\mathit{t}$, an experience tuple is recorded as \( \mathit{e^{(t)}} = (\mathit{s^{(t)}, m^{(t)}, r^{(t)}, s^{(t+1)}}) \), where $\mathit{s^{(t)}}$ is the state vector with the $\mathit{i}$th element $\mathit{s_i^{(t)}}$ corresponding to the state of agent $\mathit{i}$. Subsequently, the experiences of all agents are consolidated into a Prioritized Experience Replay (PER) buffer $\mathcal{\mathit{D}}$, as not all transitions contribute equally to learning. The buffer assigns priority to each transition based on its impact on the learning process, which is measured by the change in cumulative reward. Transitions with larger changes in reward are deemed more significant and are given higher priorities. The probability of sampling a transition is proportional to its priority, meaning transitions with higher priority are more likely to be selected for training. The degree of prioritization is controlled by a parameter that adjusts how strongly the learning process favors high-priority transitions. This approach ensures that transitions that lead to significant improvements in reward are replayed more often, thereby accelerating the learning process and enhancing overall performance \cite{enabling_robust_2022}. Training is then conducted using the mini-batch $\mathcal{\mathcal{E}}^{(t)}$ sampled from this buffer to update the network parameters.

The Q-learning algorithm seeks to discover a policy that maximizes the cumulative future reward. The Q-function represents the expected future reward for taking a specific action under policy $\mathit{\pi}$, expressed as
\begin{equation}
\mathit{Q^{\pi}(s, m)} = \mathbb{E}_{\mathit{\pi}} [\mathit{R^{(t)} | s^{(t)} = s, m^{(t)} = m}],
\end{equation}
where $\mathit{R^{(t)}}$ denotes the discounted future reward, calculated as
\begin{equation}
\mathit{R^{(t)}} = \sum_{\tau=0}^{\infty} \mathit{\alpha_D^{\tau} r^{(t+\tau)}},
\end{equation}
with $\mathit{\alpha_D} \in (0, 1]$ being the discount factor. In conventional Q-learning, a lookup table of Q-values is created, and the agent selects actions according to the $\mathit{\epsilon}$-greedy policy. Initially, Q-values are randomly assigned, and the agent chooses actions based on this policy at each time step $\mathit{t}$. The $\mathit{\epsilon}$-greedy strategy allows the agent to either exploit the best-known action with probability $1 - \mathit{\epsilon}$, or explore by selecting a random action with probability $\mathit{\epsilon}$. To encourage more exploitation over time, the value of $\mathit{\epsilon}$ is progressively decreased using a decaying $\mathit{\epsilon}$-greedy algorithm where $\textstyle \mathit{\epsilon^{(t+1)}} = (1 - \mathit{\alpha_\epsilon})^{(t)} \mathit{\epsilon^{(0)}}$, where $\mathit{\alpha_\epsilon}$ is the decay factor. 

A DQN improves upon traditional Q-learning by using deep neural networks to approximate the Q-value function, enabling it to handle high-dimensional state spaces. It incorporates a replay buffer, which stores past experiences for random sampling during training to reduce data correlation, and a fixed target network, which keeps target Q-values stable by periodically updating its parameters, ensuring more reliable learning. However, DQN often overestimates Q-values, leading to instability during training. To address this, Double Deep Q-learning (DDQN) was introduced, which decouples action selection and evaluation by using an online network for selecting actions and a target network for evaluation. Furthermore, the Dueling Q-network architecture enhances learning by separately estimating the state value and the advantage function, enabling the network to identify valuable states irrespective of the chosen action. Combining these approaches, the Dueling Double Deep Q-Network (D3QN) improves both stability and performance by addressing overestimation issues and enabling more efficient learning \cite{hessel2018rainbow}. 

\begin{algorithm}[H]
\caption{Safe D3QN-Based Algorithm}
\label{alg:proposed_alg}
\begin{algorithmic}[1]

    \Statex \textbf{Initialization:}
    \State Initialize replay buffer $\mathcal{D}$.
    \For{each node $i \in \mathcal{N}$}
        \State Initialize local, train, and target networks:  
        \Statex \hspace{1em} $\mathcal{Q}(s_i^{(t)}, m_i^{(t)}; \mathit{\Psi_{i, \text{agent}}^{(t)}})$,  
        \Statex \hspace{1em} $\mathcal{Q}(s_i^{(t)}, m_i^{(t)}; \theta_{i, \text{train}}^{(t)})$,  
        \Statex \hspace{1em} $\mathcal{Q}(s_i^{(t)}, m_i^{(t)}; \theta_{i, \text{target}}^{(t)})$.  

    \EndFor

    \Statex \textbf{Experience Collection and Replay:}
    \For{each time frame $t$}
        \For{each node $i \in \mathcal{N}$}
            \State Observe state $s_i^{(t)}$, transmit to server.
            \State Server selects $m_i^{(t)}$; teacher modifies if needed.
            \State Observe next state $s_i^{(t+1)}$, and store $e_i^{(t)}$ in $\mathcal{D}$.
        \EndFor
        \If{$|\mathcal{D}| \geq \mathcal{E}$}
            \State Sample batch $\mathcal{E}^{(t)}$ using prioritized replay.
            \State Compute D3QN loss (\ref{eq:DQN_loss}), update $\theta_{i, \text{train}}^{(t)}$ (\ref{gradient_update}).
            \State Soft update $\theta_{i, \text{target}}^{(t)}$ via (\ref{soft_update}).
            \State Broadcast $\theta_{i, \text{train}}^{(t)}$ to update $\mathit{\Psi_{i, \text{agent}}^{(t)}}$.
        \EndIf
    \EndFor

    \Statex \textbf{Training and Execution:}
    \For{each $t$}
        \For{each $i \in \mathcal{N}$}
            \State Observe state $s_i^{(t)}$, transmit to server.
            \State Select action $m_i^{(t)}$ via $\epsilon$-greedy.
            \State Teacher verifies $m_i^{(t)}$, modifies if needed.
            \State Observe next state $s_i^{(t+1)}$, and store $e_i^{(t)}$ in $\mathcal{D}$.

        \EndFor
        \State Sample $\mathcal{E}^{(t)}$, compute loss (\ref{eq:DQN_loss}).
        \State Update $\theta_{i, \text{train}}^{(t)}$ (\ref{gradient_update}).
        \State Soft update $\theta_{i, \text{target}}^{(t)}$ (\ref{soft_update}).
        \State Broadcast $\theta_{i, \text{train}}^{(t)}$ to update $\mathit{\Psi_{i, \text{agent}}^{(t)}}$.
        \State Execute the verified actions.
    \EndFor

\end{algorithmic}
\label{alg1}
\end{algorithm}

In this work, a D3QN is employed, comprising two DQNs: the target network and the training network, with parameters \(\textstyle \mathit{\theta^{(t)}_{\text{target}}} \) and \( \textstyle \mathit{\theta^{(t)}_{\text{train}}} \), respectively. The target network parameters $\textstyle \mathit{\theta^{(t)}_{\text{target}}}$ are updated using a soft update method every time frame, such that 
\begin{equation} \label{soft_update}
    \mathit{\theta^{(t)}_\text{target}} = \mathit{\alpha_S \theta^{(t)}_\text{train} + (1 - \alpha_S)\theta^{(t)}_\text{target}},
\end{equation}
where $\mathit{\alpha_S} \ll 1$ is a smoothing factor ensuring greater stability in the learning process. The least squares loss is calculated from a mini-batch of sampled transitions $\mathcal{\mathit{E}}^{(t)}$ using 
\begin{equation} \label{eq:DQN_loss}
    \mathcal{L}(\mathit{\theta^{(t)}_{\text{train}}}) = \sum_{(\mathit{s, m, r, s'}) \in \mathcal{E}^{(t)}} (\mathit{y^{(t)}_{\textsc{DQN}}(r, s')} - \mathit{Q(s, m; \theta^{(t)}_{\text{train}})})^2, 
\end{equation}
where 
\begin{equation} \label{eq:target_y}
    \mathit{y^{(t)}_{\textsc{DQN}}(r, s')} = \mathit{r + \alpha_T \max_{m'} Q(s',m'; \theta^{(t)}_{\text{target}})},
\end{equation}
with $\mathit{\alpha_T} \in (0, 1)$ as the discount factor.

The agent minimizes the loss function by adjusting the training network parameters $\textstyle \mathit{\theta_{\text{train}}^{(t)}}$ using stochastic gradient descent, selected for its rapid convergence time \cite{lecun2015deep}. The gradient update is performed as 
\begin{equation} \label{gradient_update}
    \mathit{\theta_\text{train}^{(t)}} \leftarrow \mathit{\theta_\text{train}^{(t)} - \lambda^{(t)} \nabla_{\theta_\text{train}^{(t)}} \mathcal{L}(\theta_\text{train}^{(t)})},
\end{equation} 
where $\mathit{\lambda^{(t)}} \in (0, 1)$ represents the learning rate. To ensure a more stable learning process, the learning rate is progressively decreased using \( \textstyle \mathit{\lambda^{(t+1)}} = (1 - \mathit{\alpha_L})\mathit{\lambda^{(t)}}\), where $\mathit{\alpha_L}$ is the learning rate decay factor. 

At the end of each time frame, the updated D3QN parameters $\mathit{\theta_{i, \text{train}}^{(t)}}$ are synchronized with their corresponding local D3QN models, updating the local weights $\mathit{\Psi_{i, \text{agent}}^{(t)}}$.

The overall algorithm for the proposed safe multi-agent D3QN-based resource allocation framework is summarized in Algorithm~\ref{alg1}, detailing the initialization, experience collection, and training-execution phases. In the initialization phase, the experience replay buffer \( \mathcal{D} \) is created (Line 1). Each node \( i \in \mathcal{N} \) initializes its local, train, and target D3QN networks with randomly assigned weights (Lines 2--4). During the experience collection phase, each node observes its local state and sends it to the central server (Line 7). The server selects a random blocklength, which is verified or modified by the teacher (Line 8). The resulting next state is observed, and the experience is stored in \( \mathcal{D} \) (Line 9). Once sufficient data is available, a batch is sampled to compute the loss and update the train and target networks (Lines 11--16). In the training-execution phase, each node again observes its state and communicates it to the server (Line 20). An action is selected using the \( \epsilon \)-greedy policy and then validated or adjusted by the teacher (Lines 21--22). The resulting experience is stored (Line 23). The networks are updated as in the training phase. Finally, the verified actions are executed (Lines 25--29).

\section{Performance Evaluation} \label{results}

In this section, we analyze the performance of the proposed optimization theory-based safe DRL algorithm in comparison to rule-based DRL and optimization theory-based DRL benchmarks. The rule-based DRL benchmark serves as a safe RL baseline, where system constraints are enforced through predefined rules. While the reference approach in \cite{nguyen2022fuzzy} directly encodes expert-defined rules for action selection, we adapt this concept by generating actions randomly and repeating until the selected action satisfies all constraints. The optimization theory-based DRL benchmarks employ D3QN and DDQN models without any integrated safe mechanisms.

\subsection{Simulation Setup} \label{sim_setup}
Simulations are run for a network comprising uniformly distributed nodes that are in communication with a central controller inside a 50 m radius circle. Both large-scale and small-scale fading affect the connections between the sensor nodes and the gateway. The large-scale fading $\mathit{\zeta_{i}}$ happens due to path loss and shadowing effects caused by the obstacles between the transmitter and the receiver. It is modeled as follows: $\textstyle \mathit{PL(d)} = \mathit{PL(d_0)} + 10 \mathit{\zeta}\log(\frac{\mathit{d_i}}{\mathit{d_0}}) + \mathit{Z}$ dB, where $\mathit{d_i}$ is the node's distance from the central controller, $\textstyle \mathit{PL(d_i)}$ is its path loss at $\mathit{d_i}$ from the controller, measured in decibels, $\mathit{PL(d_0)}$ = 35.3 dB is its path loss at reference distance $\mathit{d_0} = 1 $ m, and $\mathit{\zeta}=3.76$ is the path loss exponent \cite{access2010further}. $\mathit{Z}$ is the log-normal shadowing corresponding to a Gaussian random variable with zero mean and standard deviation equal to 4 dB. Jake's model, which is represented as a first-order complex Gauss-Markov process, is used for the small-scale fading:
 \begin{equation}
     \mathit{f_{i}^{(t)}} = \mathit{\rho f_{i}^{(t-1)}} + \sqrt{1 - \mathit{\rho^2}}\mathit{e_{i}^{(t)}},
 \end{equation}
where $\mathit{\rho}$ is the correlation coefficient, which is set to 0.6. Also, $\mathit{f_{i}^{(t)}}$ and $\mathit{e_{i}^{(t)}}$ are the channel coefficient and the channel innovation process of node $\mathit{i}$ at time $\mathit{t}$, respectively. $\mathit{f_{i}^{(0)}}$ and $\mathit{e_{i}^{(1)},e_{i}^{(2)}, ...}$ are independent and identically distributed circularly symmetric complex Gaussian (CSCG) random variables with unit variance. Accordingly, the channel gains are computed by
 \begin{equation}
     \mathit{g_{i}^{(t)}} = |\mathit{f_{i}^{(t)}}|^2 \mathit{\zeta_{i}}, \quad t = 1, 2, ....
 \end{equation}

The noise power spectral $\mathit{\sigma^2}$ is set to -174 dBm/Hz. The other parameters used in the simulations are given in Table \ref{table:parameters}. The simulation parameters are selected in accordance with typical URLLC settings and prior works on short-packet communication and finite blocklength analysis~\cite{aoi_survey_2021, 2022_popovski, hamida_amir_2024, sadi2014minimum}. Although the packet length $L_i$ is fixed for all nodes, the resulting blocklengths $m_i$ vary by node. This variation arises from differences in wireless channel gains, sampling intervals, and packet error probabilities, all jointly considered in the optimization framework.

\begin{table}[h!]
\centering
\caption{Simulation Parameters} 
\label{table:parameters} 
\begin{adjustbox}{width=0.45\textwidth}
\scriptsize 
\begin{tabular}{|c|c|c|c|}
    \hline
    \textbf{Parameter} & \textbf{Value} & \textbf{Parameter} & \textbf{Value} \\
    \hline
    $\mathit{B}$   & 100 kHz    & $\mathit{M_{\text{th}}}$ & 200 Symbols \\
    \hline
    $\mathit{L_i}$ & 100 bits   & $\mathit{\delta}$  & 0.99 \\
    \hline
    $\mathit{\alpha}$   & 101 ms & $\mathit{N}$ & 50 \\ 
    \hline
    $\mathit{W_{\text{max}}}$ & 250 mW & $\mathit{W_c}$  & 5 mW \\ 
    \hline
    $\mathit{\beta}$ & 0.9 & $\mathit{\sigma^2}$  & -174 dBm/Hz \\ 
    \hline
\end{tabular}
\end{adjustbox}
\end{table}

\begin{table}[h!]
\centering
\caption{Hyperparameters}
\label{table:hyperparameters} 
\begin{adjustbox}{width=0.45\textwidth}
\scriptsize 
\begin{tabular}{|c|c|c|c|c|c|}
    \hline
    \textbf{Layers} & \textbf{Input} & \( \mathit{N_1} \) & \( \mathit{N_2} \) & \( \mathit{N_3} \) & \textbf{Output} \\ 
    \hline
    \textbf{Teacher-Student} & &  &  &  &  \\ 
    \textbf{Rule-Based} & $3\mathit{N}+3$ & 32 & 64 & 300 & $\mathit{M_{\text{th}}}, 1$ \\ 
    \textbf{D3QN}  & & & & &  \\ 
    \hline
    \textbf{DDQN} & $3\mathit{N}+3$ & 32 & 64 & 300 & $\mathit{M_{\text{th}}}$ \\ 
    \hline
    \textbf{Activation Function} & Linear & ReLU & ReLU & ReLU & Linear \\ 
    \hline
    \textbf{Mini-Batch Size} & \multicolumn{5}{c|}{64} \\ 
    \hline
    \textbf{Testing} & \multicolumn{5}{c|}{10 simulations \(\times\) 2,500 episodes} \\ 
    \hline
    \textbf{Discount Factor \( \mathit{\alpha_T} \)} & \multicolumn{5}{c|}{0.666} \\ 
    \hline
    \textbf{Soft Update Rate \( \mathit{\alpha_S} \)} & \multicolumn{5}{c|}{$10^{-3}$} \\ 
    \hline
    \textbf{Initial Learning Rate \( \mathit{\lambda^{(0)}} \)} & \multicolumn{5}{c|}{0.03} \\ 
    \hline
    \textbf{Learning Decay Rate \( \mathit{\alpha_L} \)} & \multicolumn{5}{c|}{$10^{-3}$} \\ 
    \hline
    \textbf{Initial Exploration Rate \( \mathit{\epsilon^{(0)}} \)} & \multicolumn{5}{c|}{1} \\ 
    \hline
    \textbf{\(\mathit{\epsilon}\)-Decay Rate \( \mathit{\alpha_\epsilon} \)} & \multicolumn{5}{c|}{$10^{-4}$} \\ 
    \hline
\end{tabular}
\end{adjustbox}
\end{table}

The network architecture consists of one input layer, three hidden layers ($\mathit{N_1}$, $\mathit{N_2}$, $\mathit{N_3}$), and one output layer. The vanishing gradient issue is avoided by using the leaky ReLU activation function in the hidden layers. For Teacher-Student, Rule-Based, and D3QN algorithms, the final layer outputs the advantage and state value functions, featuring $\mathit{N_A}$ neurons for the advantage function, corresponding to $\mathit{M_{\text{th}}}$, and $\mathit{N_V}$ neurons for the state value function, which is set to 1. For the DDQN benchmark, the output is equal to the number of possible actions, which is $\mathit{M_{\text{th}}}$. The mini-batch size is set to 64 for all algorithms. The optimal hyperparameters are determined by employing the grid search algorithm, which systematically examines a grid of all potential combinations of hyperparameter values. The selected hyperparameters are tabulated in Table \ref{table:hyperparameters}. The proposed algorithm is implemented using PyTorch \cite{paszke2019pytorch}. Each test result is averaged over 10 simulations, with each simulation initialized using a unique random seed and spanning 2,500 episodes. 

\begin{figure*}[t] 
    \centering
  \subfloat[\label{1a}]{%
       \includegraphics[width=0.43\linewidth]{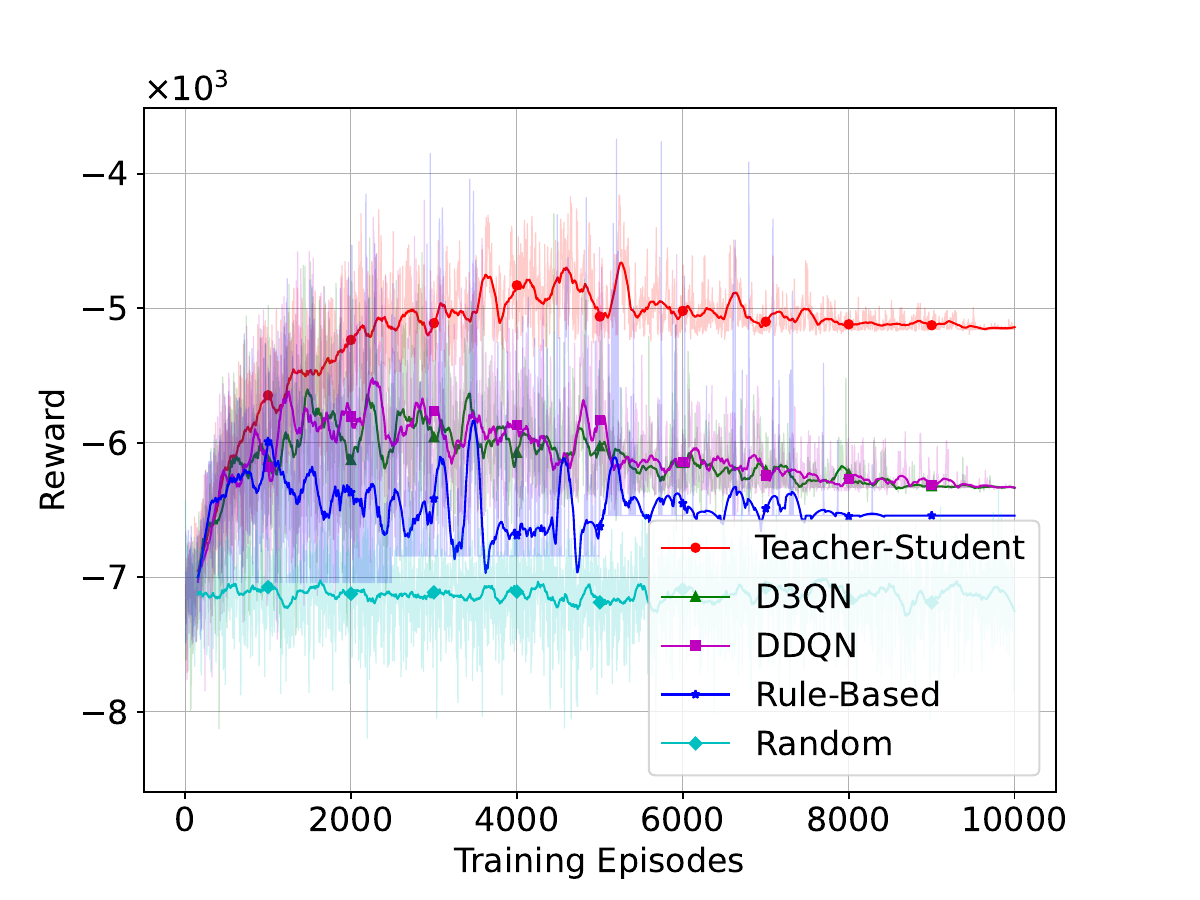}}
    \hspace{-0.5em} 
  \subfloat[\label{1b}]{%
        \includegraphics[width=0.43\linewidth]{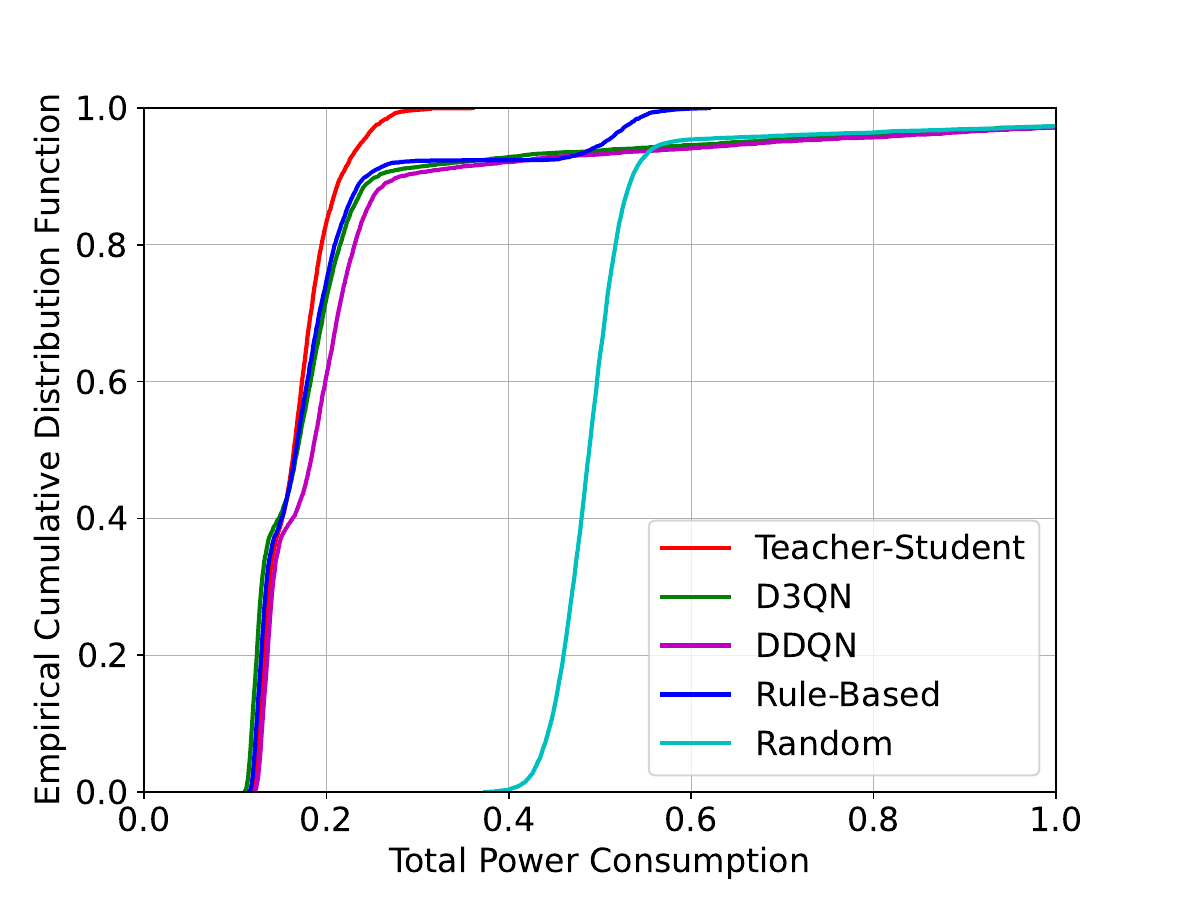}}
    \vspace{-1.3em} 
    \\
  \subfloat[\label{1c}]{%
        \includegraphics[width=0.43\linewidth]{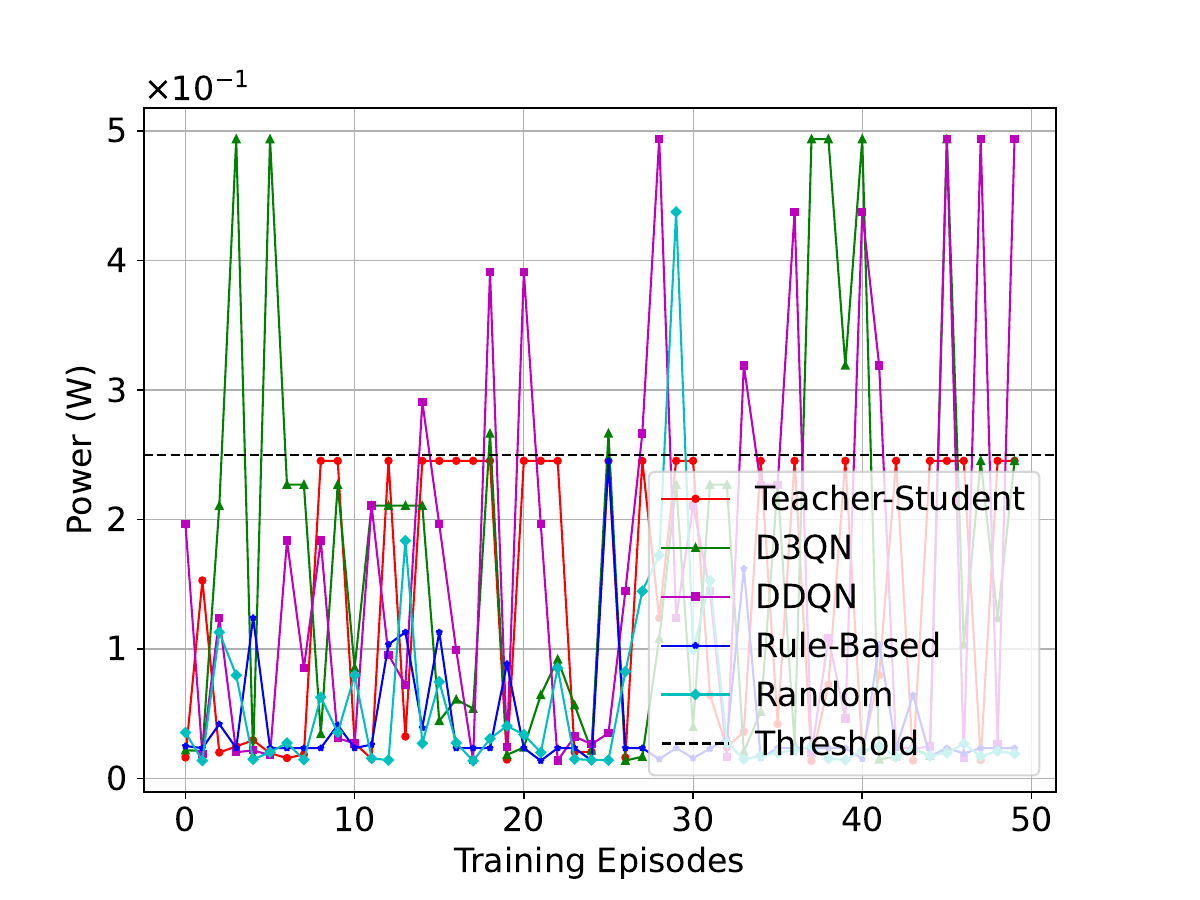}}
    \hspace{-0.5em} 
  \subfloat[\label{1d}]{%
        \includegraphics[width=0.43\linewidth]{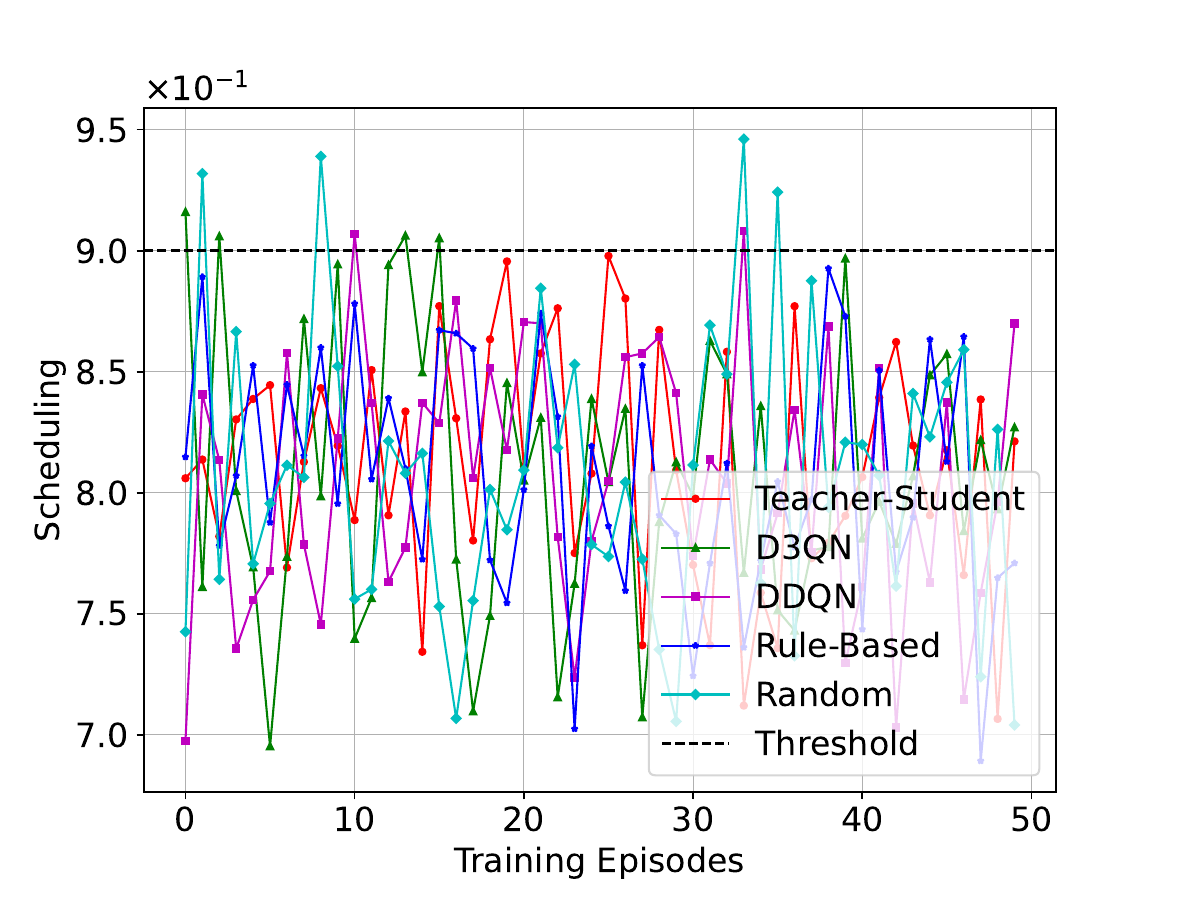}}
  \caption{(a) Training, (b) testing, (c) power violation, and (d) scheduling violation results for different algorithms in a network of 50 nodes with $\mathit{\alpha} = 101 \text{ ms}$.}
  \label{50_nodes_099} 
\end{figure*}

\begin{figure*}[t] 
    \centering
  \subfloat[\label{3a}]{\includegraphics[width=0.43\linewidth]{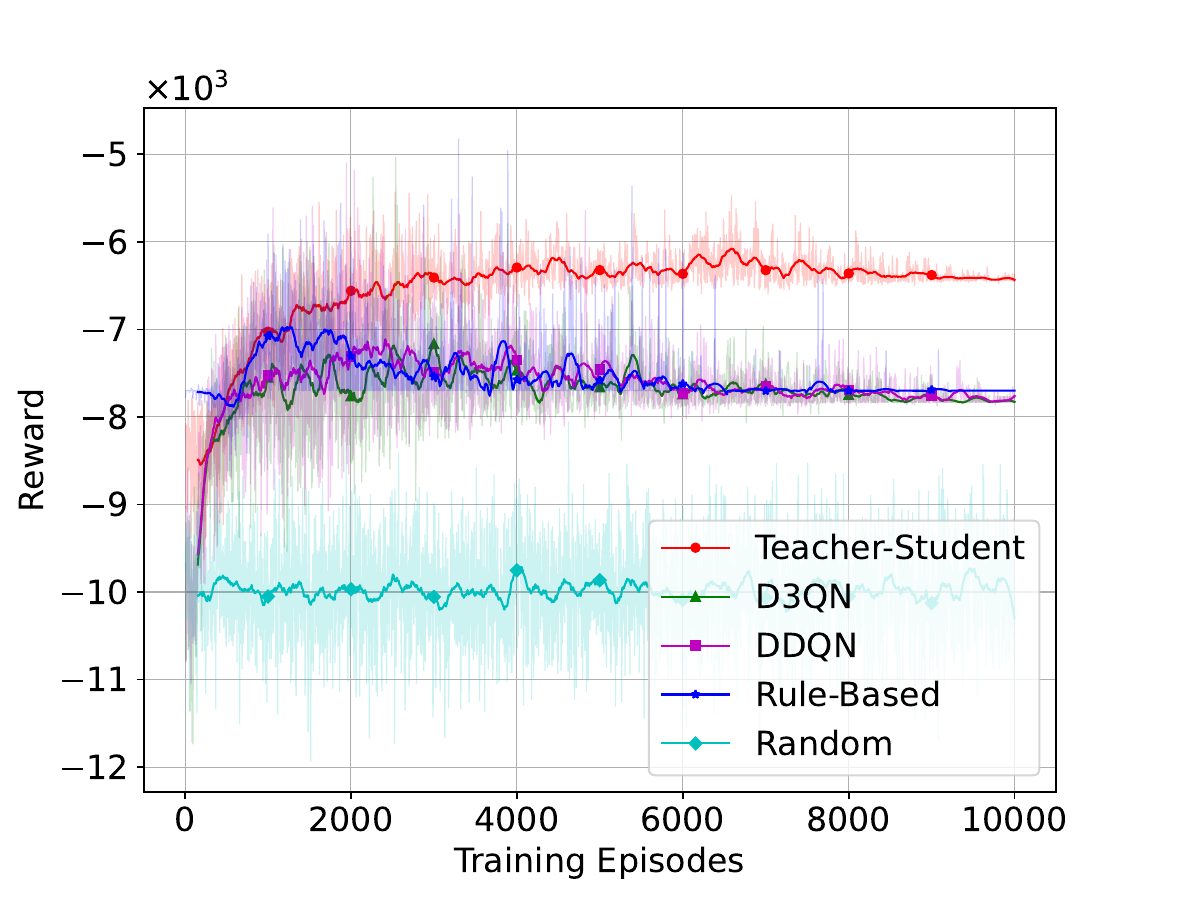}}
  \hspace{-0.5em} 
  \subfloat[\label{3b}]{\includegraphics[width=0.43\linewidth]{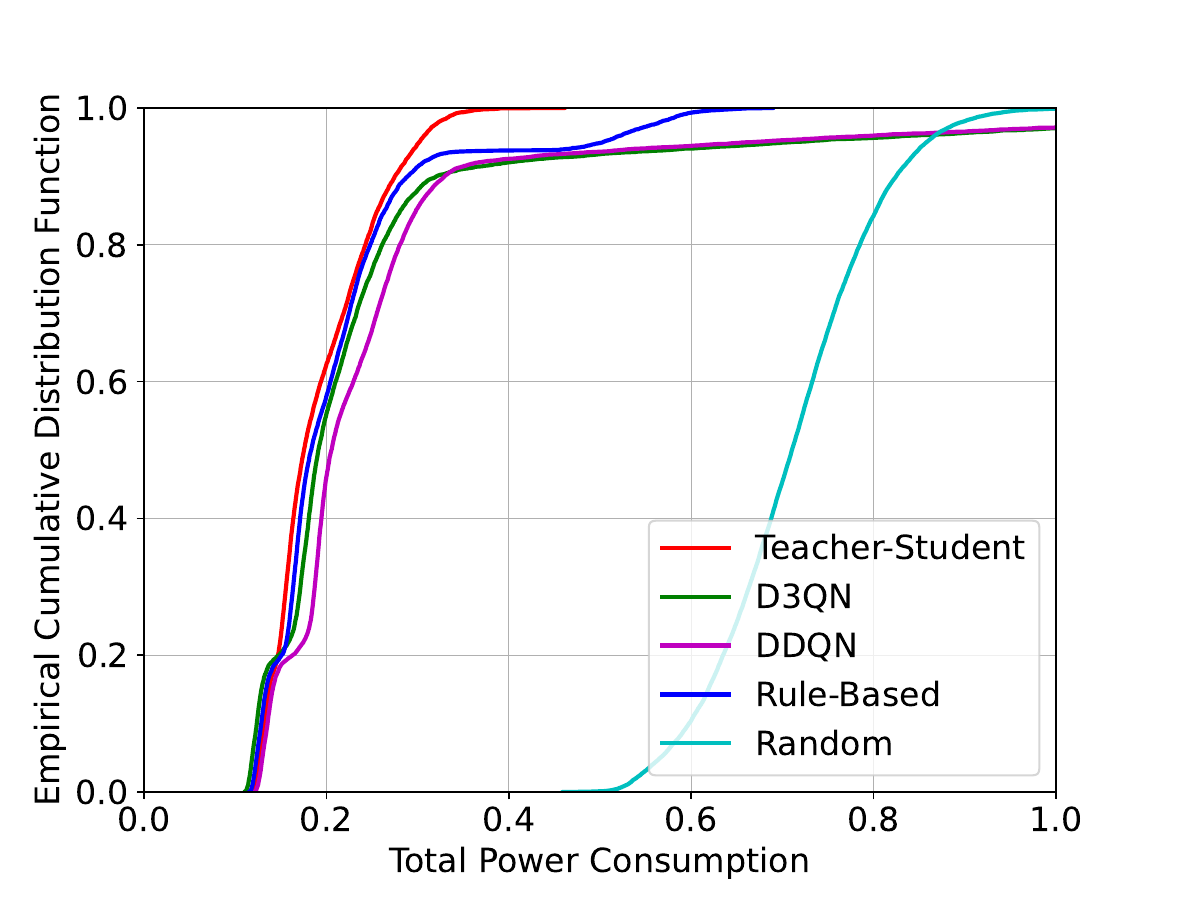}}
    \vspace{-1.3em} 
    \\
  \subfloat[\label{3c}]{\includegraphics[width=0.43\linewidth]{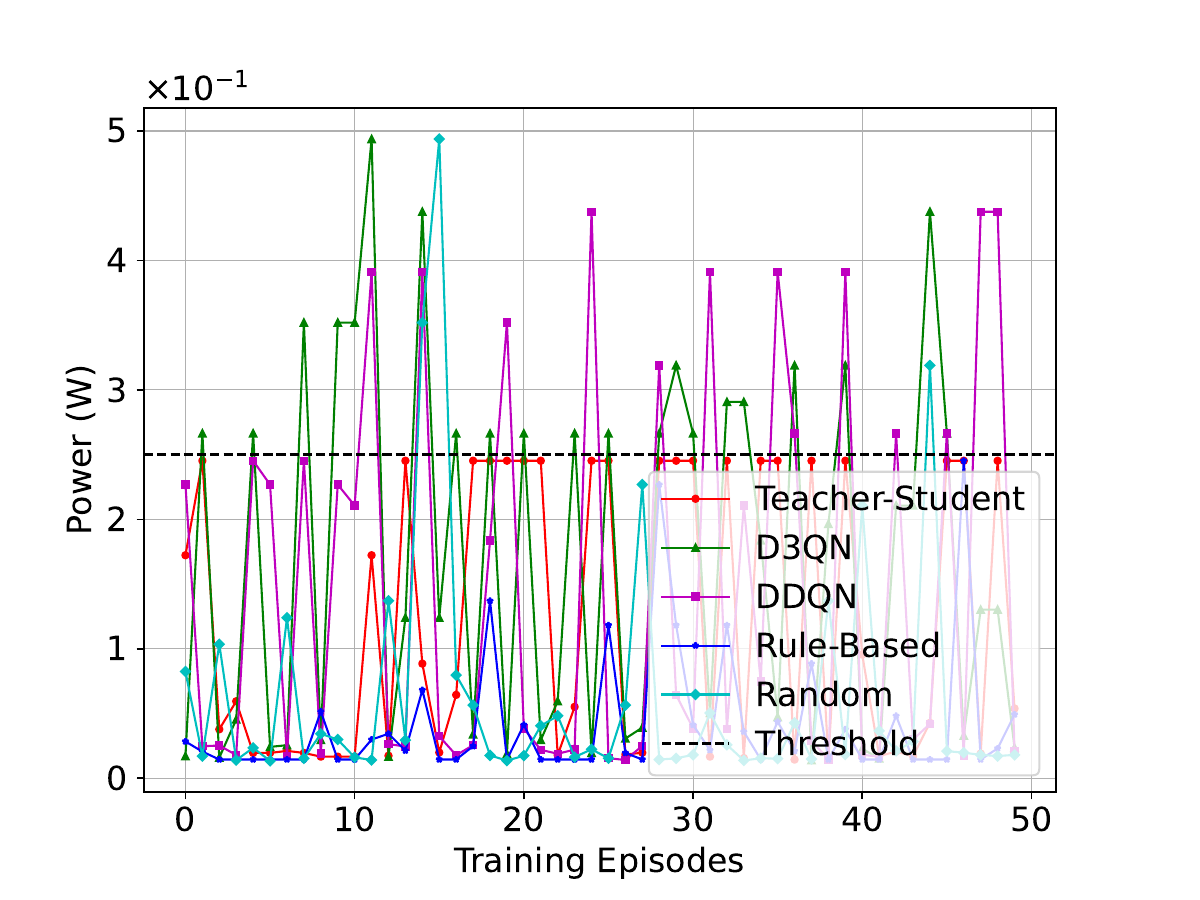}}
  \hspace{-0.5em} 
  \subfloat[\label{3d}]{\includegraphics[width=0.43\linewidth]{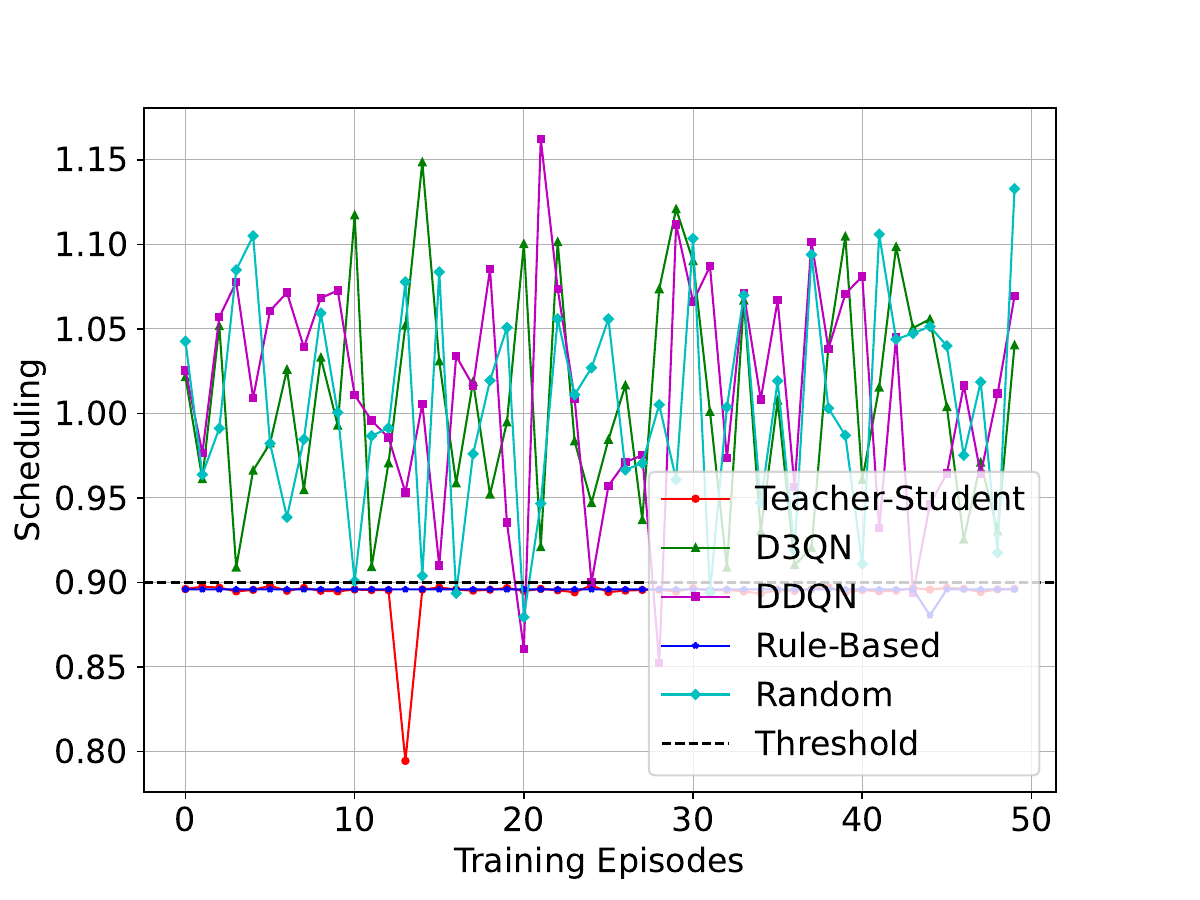}}
  \caption{(a) Training, (b) testing, (c) power violation, and (d) scheduling violation results for different algorithms in a network of 50 nodes with $\mathit{\alpha} = 81 \text{ ms}$.}
  \label{50_nodes_0081} 
\end{figure*}

\begin{figure*}[t] 
    \centering
  \subfloat[\label{4a}]{%
       \includegraphics[width=0.43\linewidth]{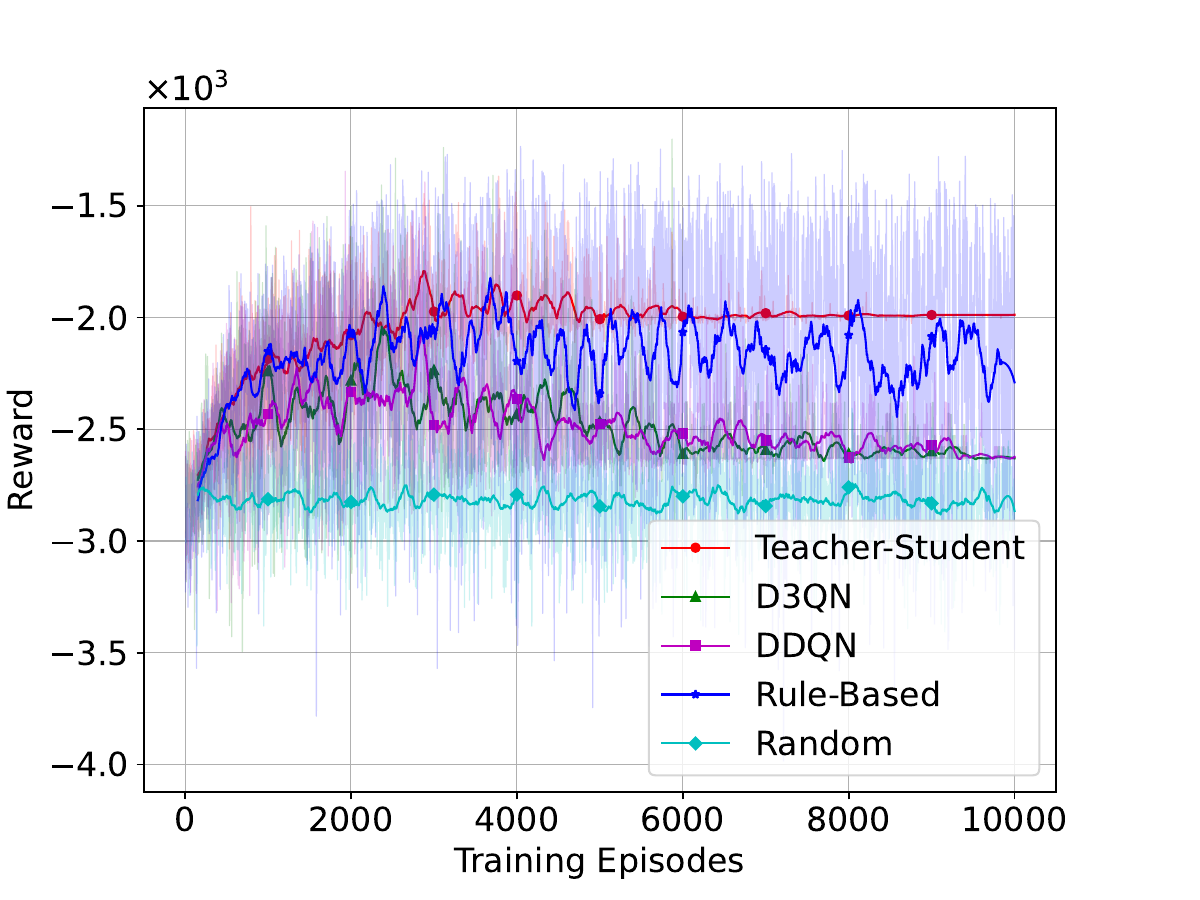}}
    \hspace{-0.5em} 
  \subfloat[\label{4b}]{%
        \includegraphics[width=0.43\linewidth]{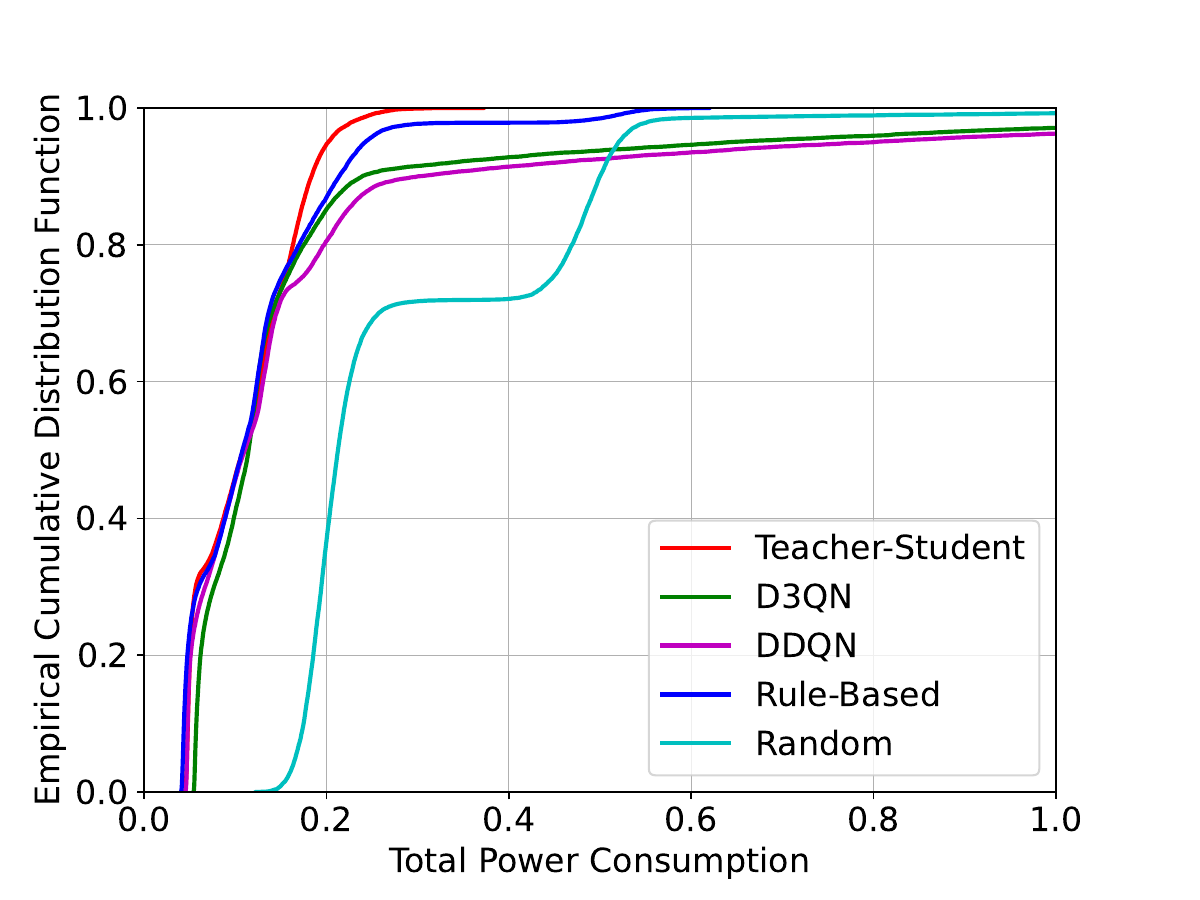}}
    \vspace{-1.3em} 
    \\
  \subfloat[\label{4c}]{%
        \includegraphics[width=0.43\linewidth]{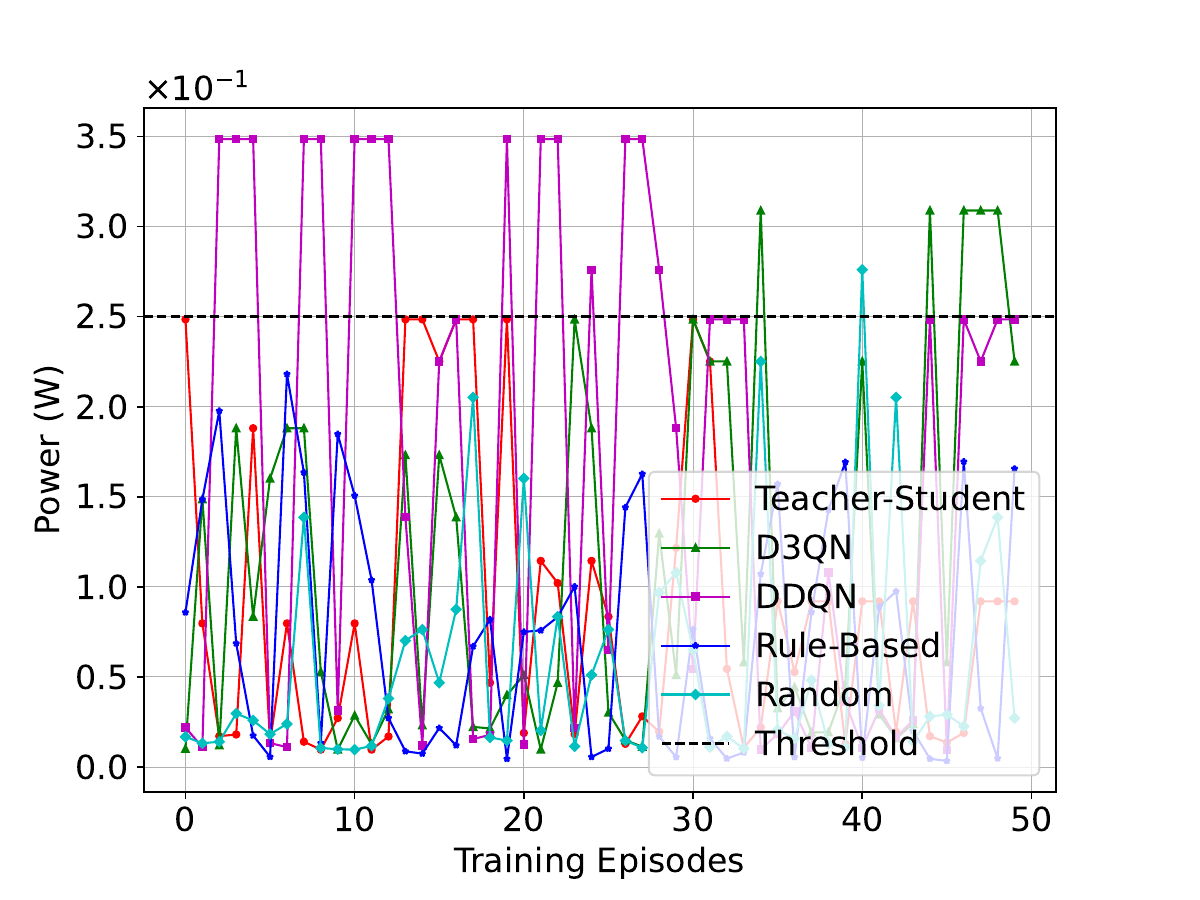}}
    \hspace{-0.5em} 
  \subfloat[\label{4d}]{%
        \includegraphics[width=0.43\linewidth]{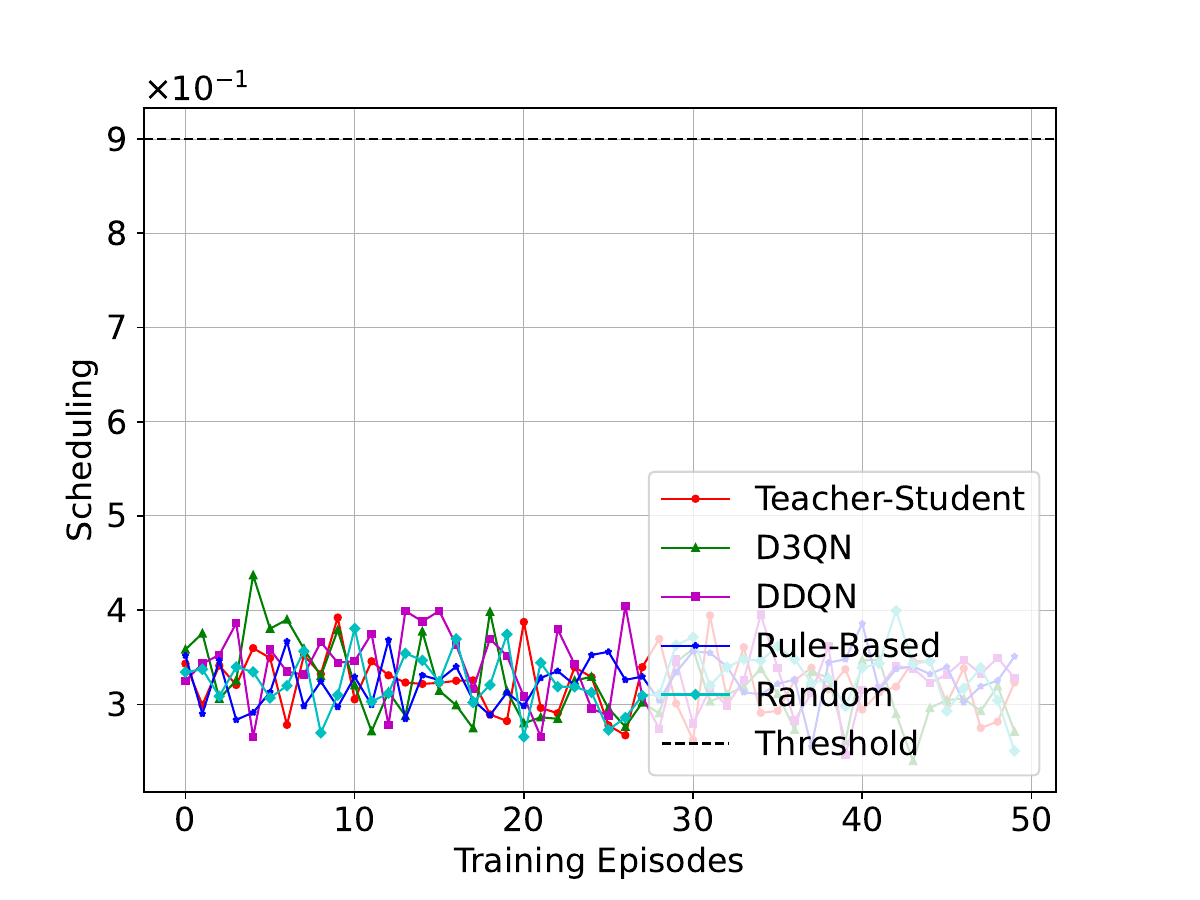}}
  \caption{(a) Training, (b) testing, (c) power violation, and (d) scheduling violation results for different algorithms in a network of 20 nodes with $\mathit{\alpha} = 101 \text{ ms}$.}
  \label{20_nodes} 
\end{figure*}
\subsection{Performance Comparison of Algorithms}
Fig.~\ref{50_nodes_099}(a) shows the training convergence of the reward for different algorithms in a network of 50 nodes with $\mathit{\alpha} = 101 \text{ ms}$. The proposed algorithm demonstrates significantly superior performance compared to benchmark algorithms, showing faster convergence, higher rewards, and greater stability. This is due to its ability to intelligently select actions that satisfy system constraints, leading to higher rewards. D3QN performs slightly better than DDQN in terms of average reward, benefiting from its robust architecture and capability to identify patterns in a stochastic environment. On the other hand, the rule-based D3QN exhibits instability as it resorts to random actions whenever a rule is violated. Fig.~\ref{50_nodes_099}(b) presents the empirical cumulative distribution function (CDF) of the total power consumption (in Watts) obtained during the testing phase using the parameters saved at the end of training. The x-axis indicates the total power consumed per episode, while the y-axis represents the corresponding cumulative probability. This visualization enables a direct comparison of the energy efficiency of the methods by illustrating how frequently each achieves lower power usage. The results align closely with the outcomes observed during training. The teacher-student method demonstrates lower power consumption compared to the benchmark algorithms, with its primary advantage being strict compliance with system constraints, an attribute not observed in the D3QN, DDQN, and random approaches. On the other hand, rule-based, D3QN, and DDQN methods exhibit similar power consumption distributions. Finally, the random approach shows the highest power consumption, highlighting its inefficiency relative to the other methods.

The power violation plot for a single node over 50 episodes is shown in Fig. \ref{50_nodes_099}(c). Both the teacher-student and rule-based algorithms consistently satisfy the power constraint for all nodes. Although the rule-based D3QN guarantees constraint satisfaction, it exhibits high variance and instability due to its random action regeneration strategy, which becomes ineffective in complex real-time scenarios. In contrast, the D3QN, DDQN, and random benchmarks repeatedly breach the power threshold for this node. To further evaluate the power violation performance across different algorithms, two metrics are used: the episode violation rate, which indicates the percentage of episodes where any node breaches the constraint, and the average node violation rate, which provides a granular view of constraint adherence across all nodes. For the D3QN algorithm, the episode violation rate is 80.7\%, with an average node violation rate of 1.72\%. Similarly, the DDQN algorithm records an episode violation rate of 83.9\% and an average node violation rate of 1.68\%. In contrast, the random approach shows significantly lower values, with an episode violation rate of 4.08\% and an average node violation rate of only 0.093\%. The high episode violation rates observed in the D3QN and DDQN algorithms primarily result from their reward design, which prioritizes minimizing the objective function over satisfying power constraints for certain nodes. In DRL frameworks, as the reward is defined as the sum of the objective and penalties for constraint violations, the agent may favor actions that optimize the objective while accepting penalties for breaches. This trade-off often causes the agent to prioritize the objective at the expense of strict adherence to the constraints. The scheduling violation plot is shown in Fig. \ref{50_nodes_099}(d). Similarly, the teacher-student and rule-based algorithms satisfy the constraint in all of the episodes. However, the D3QN, DDQN, and random benchmarks exceed the threshold, with violation rates of 0.09\%, 0.1\%, and 4.62\%, respectively. The relatively low violation probabilities of D3QN and DDQN reflect the less stringent nature of the scheduling constraint in this setup. It is important to note that the PAoI violation probability constraint is inherently satisfied across all methods, as it is enforced in the optimization theory stage. Since the optimality conditions provided in Lemmas 1 and 2 are applied to every approach, none of the tested methods violates PAoI constraints during execution. In contrast, a pure DRL approach without optimization theory would struggle to maintain this guarantee, as it lacks an explicit mechanism for enforcing constraint satisfaction. This issue was observed in \cite{hamida_amir_2024}, where pure DRL exhibited worse performance due to frequent violations. As a result, we did not include a pure DRL baseline for comparison, as it would not provide a fair evaluation under strict PAoI constraints.

\subsubsection{Impact of PAoI}
This section examines the effect of varying the PAoI threshold value, $\mathit{\alpha}$, on the performance of the algorithms during the training and testing phases. Fig. \ref{50_nodes_0081}(a) shows the training convergence of the reward for different algorithms in a network of 50 nodes with $\mathit{\alpha} = 81 \text{ ms}$. Compared to the earlier scenario, the stricter PAoI violation probability and scheduling constraints result in higher penalties. Consequently, the rule-based algorithm achieves better performance than the DDQN, DQN, and random benchmarks. However, the proposed teacher-student approach outperforms all algorithms in terms of both convergence and average reward, thanks to its logical and safe decision-making mechanism. Fig. \ref{50_nodes_0081}(b) shows the testing phase results for the algorithms. Under the strict PAoI condition, the teacher-student and rule-based approaches significantly outperform the DDQN, DQN, and random methods, as meeting the requirements becomes both more challenging and more rewarding. This result aligns with the observations from the training convergence graph. As anticipated, the stringent PAoI constraint also leads to an increase in total power consumption to ensure successful transmissions. Overall, the proposed teacher-student method once again demonstrates superior performance during the testing phase, surpassing all other algorithms.

Fig. \ref{50_nodes_0081}(c) and Fig. \ref{50_nodes_0081}(d) indicate that D3QN, DDQN, and random approaches frequently violate power and scheduling constraints, exceeding the specified thresholds. Specifically, the episode violation percentages are 87.4\%, 87.1\%, and 4.2\% for the D3QN, DDQN, and random approaches, respectively, with average node violation rates of 1.75\%, 1.74\%, and 0.089\%. Both D3QN and DDQN often exceed the maximum transmit power limit of certain nodes to optimize overall system performance. In terms of scheduling violations, the violation percentages are 2.6\%, 2.3\%, and 96\% for D3QN, DDQN, and random benchmarks, respectively. The random method exhibits a high violation rate due to the majority of its action combinations failing to meet the constraints. Conversely, D3QN and DDQN gradually learn the constraints over time, resulting in fewer violations as the high penalties discourage constraint breaches. Nonetheless, the rule-based and teacher-student methods consistently satisfy both power and scheduling requirements, even under stricter conditions, demonstrating their robustness and reliability in maintaining compliance with the increasingly challenging constraints.

\subsubsection{Impact of the Number of Nodes}
We analyze the effect of the number of nodes on the performance of the proposed algorithms. Fig. \ref{20_nodes}(a) shows the training convergence of the reward for a network of 20 nodes with $\mathit{\alpha} = 101 \text{ ms}$. As in previous setups, the proposed teacher-student framework outperforms others in terms of stability, convergence, and average reward. However, total power consumption results display greater fluctuations as the number of nodes decreases. This is because, in more relaxed scenarios, intricate approaches like the teacher-student framework may overfit during training, hindering optimal outcomes during testing. The rule-based method shows an unstable reward curve due to its reliance on random action selection, emphasizing the importance of action advice from a control mechanism. Meanwhile, DDQN and DQN benchmarks exhibit negligible differences in performance. Fig. \ref{20_nodes}(b) presents the testing outcomes using the saved training parameters. Except for the random method, all algorithms demonstrate similar performance in total power consumption due to the relaxed constraints. Teacher-student, rule-based, DDQN, and DQN methods consume less total power than the random benchmark. However, the teacher-student framework achieves the best results. In general, total power consumption decreases as the number of nodes is reduced.

As shown in Fig. \ref{20_nodes}(c), only teacher-student and rule-based methods consistently satisfy the power constraint. In contrast, the D3QN, DDQN, and random approaches exhibit episode violation rates of 77.7\%, 77.9\%, and 1.75\%, respectively, with corresponding average node violation rates of 3.89\%, 4.01\%, and 0.90\%. Similar to earlier scenarios, D3QN and DDQN fail to adequately prioritize constraints in the reward design, leading to higher violation rates. The reduced number of nodes results in lower episode violation rates, as the likelihood of encountering a violation in an episode decreases. However, average node violation rates increase because each node's behavior has a more significant impact on the overall metric. With only 20 nodes, the scheduling requirement becomes more relaxed and easier to satisfy, as reflected in Fig. \ref{20_nodes}(d). Unlike previous cases, none of the algorithms violate the scheduling constraint due to the less stringent setup. However, D3QN, DDQN, and random approaches lack safety mechanisms, which could lead to constraint violations under different parameter settings.

\subsubsection{Runtime Performance}
Fig. \ref{complexity} shows the average simulation time per iteration, measured in milliseconds, for each algorithm. The execution times of all algorithms increase nearly linearly with the number of nodes. As expected, the random approach exhibits the lowest complexity. DDQN follows with slightly higher complexity due to its architecture. D3QN and the rule-based benchmarks require more time as they introduce additional processing steps. The proposed teacher-student method incurs a marginal increase in simulation time due to its advice mechanism. However, this minor overhead is justified by its faster convergence and stable system performance, making it highly suitable for real-time applications.

\begin{figure}[htbp]
\centerline{\includegraphics[width = .4\textwidth]{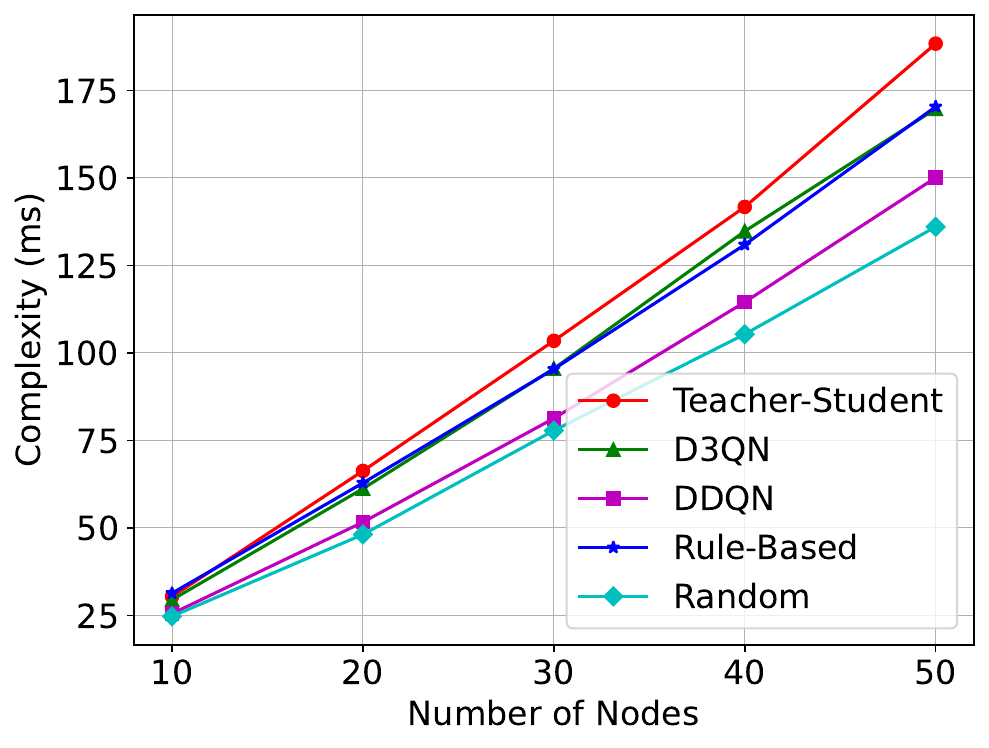}}
\caption{Average execution time (ms) per step for different algorithms and number of nodes.}
\label{complexity}
\end{figure}

\section{Conclusion} \label{conclusion}
This paper introduces a novel optimization theory-based safe DRL approach for the joint optimization of control and communication systems, incorporating a uniquely derived PAoI violation probability constraint. After formulating the joint optimization problem, we derive the optimality conditions to simplify and decompose the problem. The simplified problem is then solved using D3QN within a teacher-student learning framework, where the teacher guides the student by providing action advice that aligns best with the student's actions while satisfying system constraints. To benchmark our approach, we utilize rule-based D3QN, D3QN, DDQN, and random methods. The constraint violation percentages for D3QN, DDQN, and random approaches are presented for different parameter settings, showing that D3QN and DDQN exhibit similar violation rates, indicating that improvements in architecture have minimal impact on constraint violations. In contrast, the rule-based D3QN consistently satisfies the constraints but exhibits unstable performance due to its random advice mechanism. Our proposed algorithm outperforms the benchmarks across various node configurations and PAoI violation probability constraints, all while adhering to system requirements with only a negligible increase in computational complexity. The algorithm demonstrates stable performance due to the strategic advice mechanism integrated with the D3QN model’s action selection behavior. As future work, we aim to enhance the algorithm by incorporating advanced learning techniques such as meta-learning and transfer learning to improve the adaptability of the safe DRL method to new environments, reducing the need for extensive retraining.

\bibliographystyle{IEEEtran}


\begin{IEEEbiography}[{\includegraphics[width=1in,height=1.25in,clip,keepaspectratio]{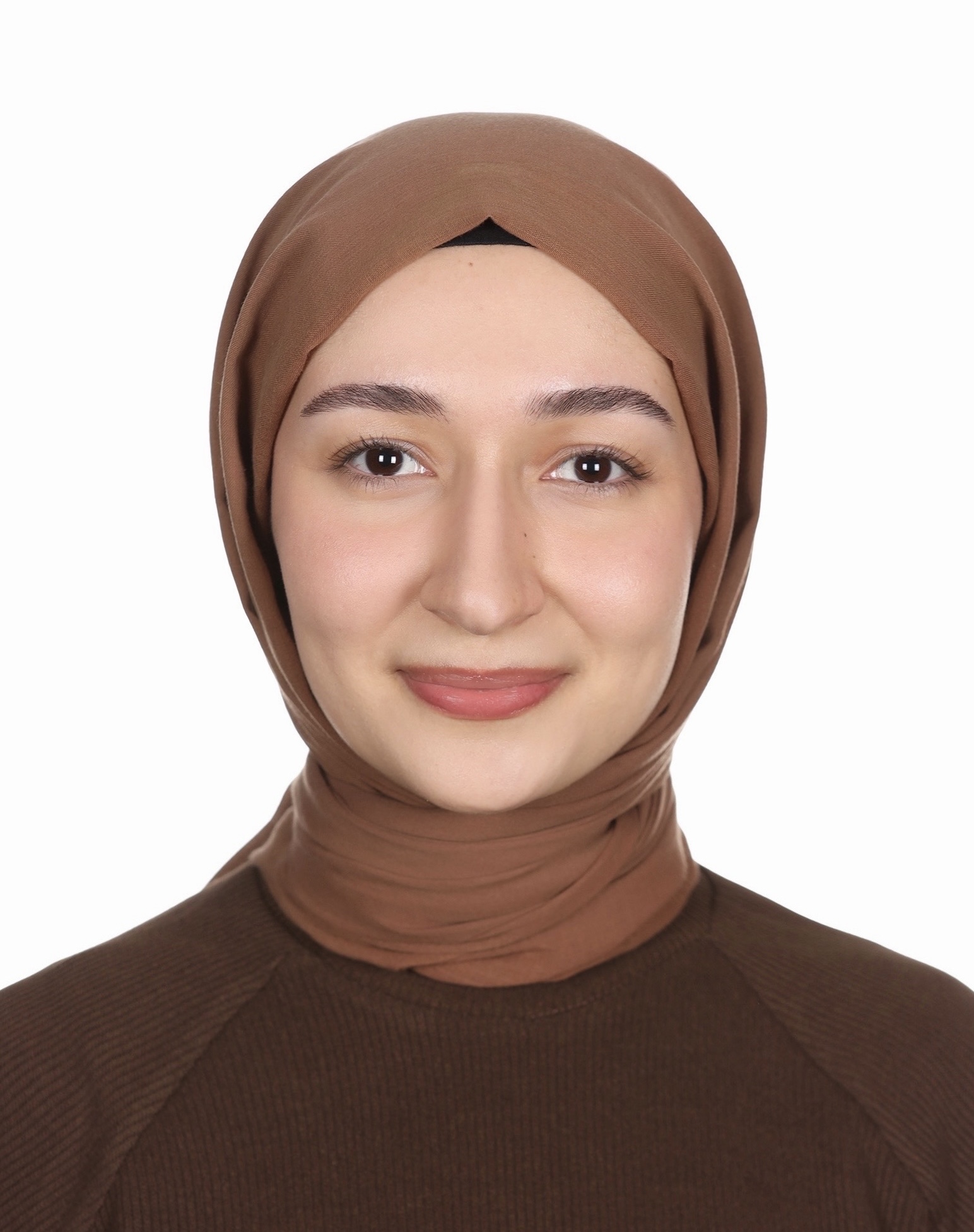}}]{Berire Gunes Reyhan}
received the B.Sc. degree in electrical and electronics engineering from Istanbul Medipol University, Istanbul, Turkey, in 2022. She is currently pursuing the M.Sc. degree in the Department of Electrical and Electronics Engineering at Koc University, Istanbul, Turkey. She is also a Research Assistant with the Wireless Networks Laboratory (WNL). Her research interests include 6G wireless communications and networking, machine learning for wireless networks, and wireless networked control systems.
\end{IEEEbiography}

\begin{IEEEbiography}[{\includegraphics[width=1in,height=1.25in,clip,keepaspectratio]{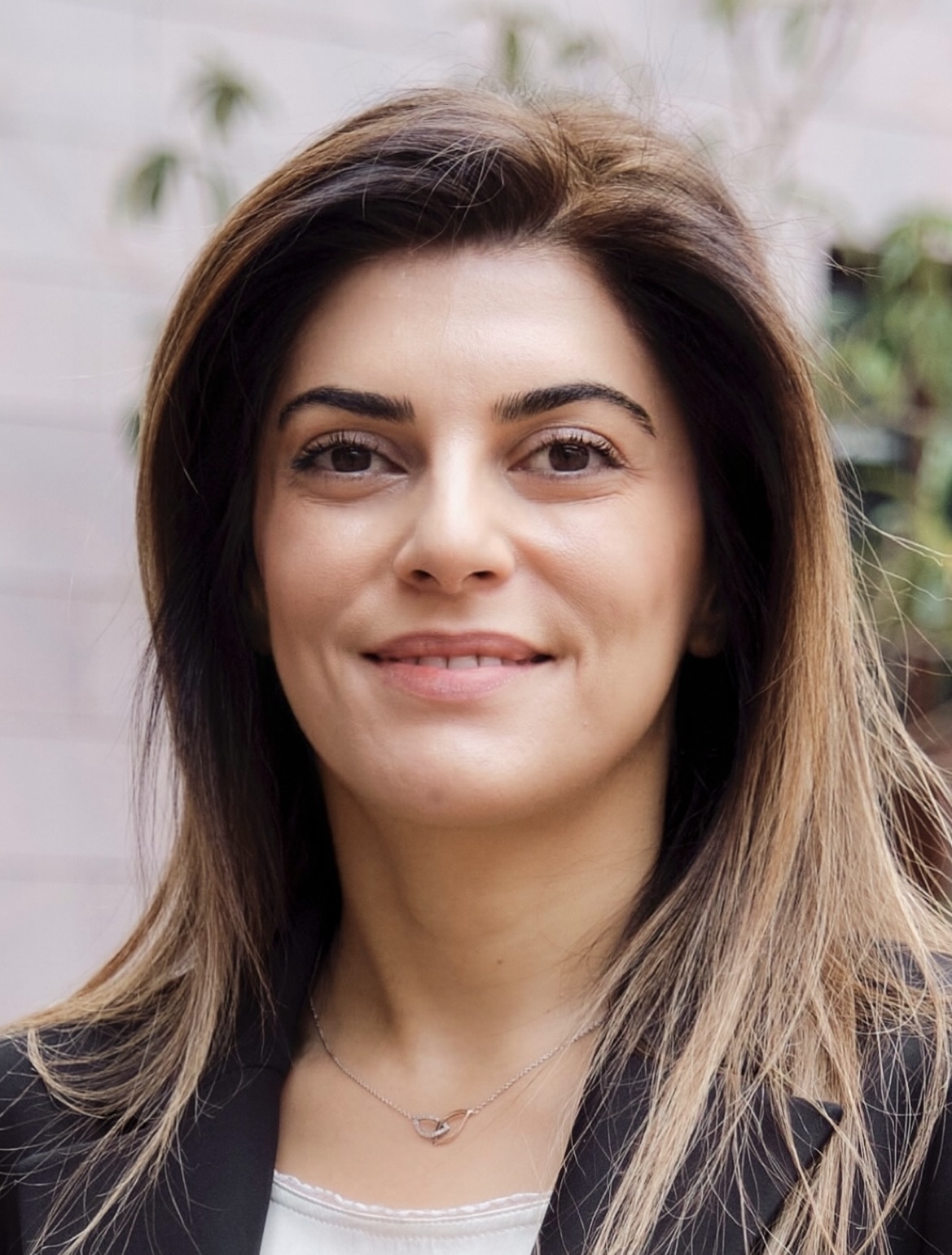}}]{Sinem Coleri}
 is a Professor in the Department of Electrical and Electronics Engineering at Koc University. She is also the founding director of Wireless Networks Laboratory (WNL) and director of Ford Otosan Automotive Technologies Laboratory. Sinem Coleri received the BS degree in electrical and electronics engineering from Bilkent University in 2000, the M.S. and Ph.D. degrees in electrical engineering and computer sciences from University of California Berkeley in 2002 and 2005. She worked as a research scientist in Wireless Sensor Networks Berkeley Lab under sponsorship of Pirelli and Telecom Italia from 2006 to 2009. Since September 2009, she has been a faculty member in the department of Electrical and Electronics Engineering at Koc University. Her research interests are in 6G wireless communications and networking, machine learning for wireless networks, machine-to-machine communications, wireless networked control systems and vehicular networks. She has received numerous awards and recognitions, including TUBITAK (The Scientific and Technological Research Council of Turkey) Science Award in 2024; N2Women: Stars in Computer Networking and Communications in 2022; TUBITAK Incentive Award and IEEE Vehicular Technology Society Neal Shepherd Memorial Best Propagation Paper Award in 2020; Outstanding Achievement Award by Higher Education Council in 2018; and Turkish Academy of Sciences Distinguished Young Scientist (TUBA-GEBIP) Award in 2015.  Dr. Coleri currently holds the position of Editor-in-Chief at the IEEE Open Journal of the Communications Society.  Dr. Coleri is an IEEE Fellow, AAIA Fellow and IEEE ComSoc Distinguished Lecturer.
\end{IEEEbiography}

\end{document}